\documentclass[10pt,journal]{IEEEtran}

\usepackage{mathrsfs}
\usepackage{amsfonts}
\usepackage{amsmath}
\usepackage{graphicx}

\onecolumn
\newtheorem{theorem}{Theorem}

\newtheorem{claim}{Claim}
\newtheorem{lemma}[theorem]{Lemma}
\newtheorem{corollary}[theorem]{Corollary}
\newtheorem{remark}{Remark}

\newtheorem{definition}{Definition}

\newcommand{\suppress}[1]{}

\newcommand{\comts}[1]{}

\newcommand{\bl}{n}            
\newcommand{\spc}{{\cal S}}
\newcommand{\net}{{\cal N}}

\newcommand{\bX}{{\bf X}}

\newcommand{\cA}{{\cal A}}

\newcommand{\F}{F}

\newcommand{\cN}{\mathcal{N}}

\newcommand{\iden}{I}

\newcommand{\bR}{{\bf R}}

\newcommand{\cC}{\mathcal{C}}

\newcommand{\cP}{\mathcal{P}}

\pubid{0000--0000/00\$00.00˜\copyright˜2011 IEEE}

\begin{document}

\title{On network coding for sum-networks}
\author{Brijesh~Kumar~Rai,~\IEEEmembership{Member,~IEEE}, and Bikash~Kumar~Dey,~\IEEEmembership{Member,~IEEE}
\thanks{The material in this paper was presented in parts
at the 2009 Workshop on Network Coding, Theory and Applications, Lausanne, Switzerland, June 2009 and at the 2009 IEEE International Symposium on Information Theory, Seoul, Korea, June 2009.}
\thanks{B.~K.~Rai was with the Department of Electrical Engineering, Indian
Institute of Technology Bombay, India. He is now with the Department
of
Electronics and Electrical Engineering, Indian Institute of Technology
Guwahati, India (email: bkrai@iitg.ernet.in). B.~K.~Dey is with the
Department of Electrical Engineering, Indian Institute of Technology
Bombay, India (email: bikash@ee.iitb.ac.in).}
\date{\today}
}

\maketitle

%
%
%
%

\thispagestyle{empty}

\begin{abstract}
A directed acyclic network is considered where all the terminals need to
recover the sum of the symbols generated at all the sources. We call such a network 
a \textit{sum-network}. It is shown that there exists a solvably (and linear 
solvably) equivalent sum-network for any multiple-unicast network,
and thus for any directed acyclic communication network. It is also
shown that there exists a linear solvably equivalent multiple-unicast
network for every sum-network. 
It is shown that for any
set of polynomials
having integer coefficients, there exists a sum-network which is scalar 
linear solvable over a finite field $F$ if and only if the polynomials have 
a common root in $F$. For any finite or cofinite set of prime numbers,
a network is constructed which has a vector linear solution of any length
if and only if the characteristic of the alphabet field is in the given
set. The insufficiency of linear network coding 
and unachievability of the network coding capacity are proved for sum-networks
by using similar known results for communication networks.
Under fractional vector linear network coding, a sum-network and its
reverse network are shown to be equivalent. However, under non-linear
coding, it is shown that there exists a solvable sum-network whose reverse
network is not solvable.
\end{abstract}
\begin{keywords}
Network coding, function computation, solvability, polynomial equations, linear network coding.
\end{keywords}

\section{Introduction}
\label{sec:intro}
\PARstart{T}{raditionally} transmission of information has been considered similar
to flow of commodity so that information is stored and forwarded
by the intermediate nodes. The seminal work by Ahlswede
et al. (\cite{ahlswede1}) showed that mixing or coding information
at the intermediate nodes in a network may provide better throughput. 
Coding at the intermediate nodes has since been referred as network coding.
For a multicast network, where the symbols generated at a single source
are to be communicated to a set of terminals, it was shown that the capacity
under network coding is the minimum of the min-cuts between the source
and each of the terminals~(\cite{ahlswede1}).

If the nodes perform linear operations,
then the code is called a linear network code.
It was shown by Li et al. (\cite{li1})
that linear network coding is sufficient
to achieve the capacity of a multicast network.
Koetter and M\'{e}dard (\cite{koetter2}) proposed an algebraic
formulation of the linear network code design problem and related
this network coding problem to finding roots of a set of polynomials.
Jaggi et al. (\cite{Jaggi1}) gave
a polynomial time algorithm for solving this problem
for a multicast network. Ho et al.~(\cite{ho1})
showed that even when the local coding coefficients
are chosen randomly and in a distributed fashion, the multicast
capacity can be achieved with high probability.

More general forms of network codes have been considered in the literature
to allow more flexibility. If a block or vector of $n$ symbols is treated
as a single symbol at all the nodes and edges, then a node
takes a function of the incoming vectors to construct an outgoing
vector on an edge. In each session,
the receivers recover $n$-length vectors of symbols from their desired
sources. Such a code is called a {\it vector network code}.
{\em Fractional network code} is an even more general form of network
code where the number of source symbols $k$ encoded in a block is possibly different
from the block length $n$. Such a $(k,n)$ code communicates $k$ source symbols
in $n$ use of the network, thus achieving the rate $k/n$ per network
use. 
\suppress{
\begin{definition}
A rate $r$ is said to be {\em achievable} under a class of network
codes if there exists a $(k,n)$ fractional network coding solution in that
class such that $k/n \geq r.$
\end{definition}
\begin{definition}
The {\it coding capacity} of a network with respect to
a class of network codes is defined to be the supremum
of the achievable rates under that class of network codes.
The term {\em coding capacity} refers to the capacity under
the class of all network codes.
When restricted to the class of linear network codes, the corresponding
capacity is called the {\em linear coding capacity}.
Dougherty et al. (\cite{dougherty4}) showed that
there exists an instance of network coding problem where network coding
capacity is unachievable.
\end{definition}
}

\pubidadjcol
A network is said to be solvable over an alphabet if there is a
rate-$1$ network code over that alphabet which satisfies the demands of the
terminals. For an underlying alphabet field,
a network is said to be $k$-length vector linear solvable
(resp. scalar linear solvable) if there is a $k$-length
vector (resp. scalar) linear network coding solution for the network.
Two networks $\net_1$ and $\net_2$ are said to be solvably equivalent
(resp. linear solvably equivalent) if $\net_1$ is solvable (resp.
solvable with linear codes) over some finite alphabet if and only if $\net_2$
is solvable (resp. solvable with linear codes) over the same alphabet.

Much of the subsequent work considered more general
networks than multicast networks. Here a terminal requires the symbols
generated at a subset of the sources.
We call such networks
where the aim is to communicate the source symbols to the appropriate
terminals, as opposed to computation of functions of the source
symbols at the terminals (which is the subject of this paper,)
as {\it communication networks}.
Designing a network code
for a general communication network was shown to be an NP-hard 
problem~(\cite{lehman, langberg2}). It was shown in \cite{medard,riis}
that for communication
networks, scalar linear network coding is not sufficient in the sense
that there may exist a vector linear network code though no scalar linear
network code exists. Even vector linear network codes were
shown to be insufficient by Dougherty et al.~(\cite{dougherty1}).
They further showed in \cite{dougherty2} that for every set of
integer polynomial equations, there exists a directed acyclic network
which is scalar linear solvable over a finite field if and only if the set of
polynomial equations has a solution over the same finite field.
\label{page:communication_network}


\label{page:MUN}
A network with some source-terminal pairs where every terminal
wants the data generated at the respective source is called
a multiple-unicast network. It was shown in \cite{dougherty3} that from
any communication network, one can construct
a solvably equivalent multiple-unicast network.
For a multiple-unicast network, the {\em reverse network}
is obtained by reversing the direction of the edges and interchanging
the roles of the sources and the terminals. 
It was shown in  \cite{koetter3,riis1} that a multiple-unicast
network is linear solvable if and only if its reverse
network is linearly solvable.
However, it was shown in \cite{riis1,dougherty3} that there exists
a multiple-unicast network which is solvable by a non-linear network
code though its reverse network is not solvable.

In this paper, we consider a network with some sources $s_1,s_2,\cdots,s_m$ 
and some terminals $t_1,t_2,\cdots, t_n$
where each terminal requires the sum $X_1+X_2+\cdots + X_m$ of the symbols
$X_1,X_2,\cdots, X_m$ generated at the sources. We call such a network a {\it sum-network}.
In general, the problem of communicating a function
of the source symbols to the terminals is of considerable interest,
for example, in sensor networks. Our aim in this paper is to
investigate a simple but illustrative case of this problem.
Computation of the `sum' function is also of practical interest, and it
has been considered in various setups
in \cite{korner1,krithivasan1,gallager2,ramamoorthy}.
In \cite{krithivasan1}, a source coding technique was proposed for arbitrary
functions using codes for communicating linear functions. 
We consider the problem in the setup first considered by
Ramamoorthy~(\cite{ramamoorthy}).
He showed that if the number of sources or the number of terminals in
the network is at most two, then all the terminals
can compute the sum of the source symbols available at the sources using
scalar linear network coding if and only if every source node is connected
to every terminal node. There have been some work in parallel to the present
work. Langberg et al.~(\cite{langberg3}) showed that for a directed
acyclic sum-network having $3$ sources and $3$ terminals and every
source connected with every terminal by at least two distinct paths,
it is possible to communicate the sum of the sources using scalar
linear network coding. Appuswami et al.~(\cite{appuswamy1, appuswamy2})
considered the problem of communicating more general functions, for
example, arithmetic sum (unlike modulo sum as in a finite field), to
one terminal. It is worth mentioning some recent work which have
been published or become available during the review process of this paper.
A simpler proof of the sufficiency of $2$-connectedness for the solvability
of a sum-network with $3$ sources and $3$ terminals is presented in
\cite{langberg4,langberg5}.
A necessary and sufficient condition and also many simpler sufficient
conditions for solvability of sum-networks with $3$ sources and
$3$ terminals are presented by Shenvi and Dey~(\cite{shenvi4,shenvi3}).
Appuswami et al.~(\cite{appuswamy3}) presented cut-based bounds on the capacity
of computing arbitrary functions over arbitrary networks and studied
the tightness of these bounds.

The problem of distributed function computation, also 
known as in-network computation, has been of significant interest,
and the problem has been addressed in various other contexts
and setups in the past.
There are information theoretic works with the aim of characterizing
the achievable rate-region of various small and simple networks with
no intermediate nodes, but with correlated
sources~(\cite{korner1,han1,orlitsky1,feng1,krithivasan1}).
A large body of work exists for finding the asymptotic scaling laws for
communication requirements in computing functions over large networks
(see, for example, \cite{tsistsiklis,gallager2,giridhar,kanoria,boyd}).

\suppress{
\begin{definition}
A directed network with some sources and some terminals
where each source generates possibly multiple
independent random processes and each terminal requires to recover
a function of the source random processes is called a {\bf Type II}
network.
\end{definition}
Though Type II network is defined here for completeness, we will
mostly deal with simpler Type II networks in this paper.
\begin{definition}
A Type II network where every source generates one random process
and all the terminals require to recover the same function of the
random processes is called a {\bf Type IIA network}.
\end{definition}
As a further special case, we will consider a network where the sources
generate random processes over a module $M$ over a commutative ring $R$,
and the terminals
demand a linear function of the processes of the form
\begin{eqnarray}
f(X_1,X_2,\ldots,X_m) = a_1X_1 + a_2X_2 +\ldots + a_mX_m , \nonumber
\end{eqnarray}
where $a_1,a_2, \ldots, a_m \in R$.
Such a network will be called a {\bf linear-network}.
The special case of $a_1=a_2=\ldots = a_m = 1$ when $R$ is a commutative
ring with identity will be called a {\bf sum-network}.
Note that, since any abelian group is a module over the integer ring,
the sum-network is naturally defined over an abelian group.
An abelian group structure is also necessary in the
alphabet for a clean definition of the problem.
}

We assume that the source alphabet is an abelian group so that
the `sum' of the source symbols is well-defined.
If the alphabet has a field structure,
or in general a module structure over
a commutative ring with identity, then linear codes over the underlying field
or ring can be used~(\cite{dougherty1}). For linear network codes, we assume the alphabet
to be a field and mention whenever a result holds for more general
alphabets. For non-linear codes, no structure other than the abelian
group structure is assumed on the alphabet.
We assume that
the links in the network can carry one symbol from the alphabet
per use in an error-free and delay-free manner.

In the following, we outline the contributions and organization of this paper.

\subsection{Contributions and organization of the paper}
In Section~\ref{sec:back}, we present the system model and some definitions.
\begin{itemize}
\item In Section~\ref{con_networks}, for any directed acyclic multiple-unicast network, we construct a sum-network
which is solvable (resp. $k$-length vector linear solvable) if and
only if the original multiple-unicast network is solvable (resp. $k$-length
vector linear solvable), see Theorem~\ref{thm:SN_using_MUN}.
Further, the reverse of the equivalent sum-network is also
solvably equivalent to the reverse of the multiple-unicast network (Theorem~\ref{thm:SN_using_MUN}). We prove that the capacity of the constructed
sum-network is lower bounded by that of the multiple-unicast network
(Theorem~\ref{th:capacity}).
For any given sum-network, we then construct a multiple-unicast
network which is linear solvably equivalent to the given sum-network
(Theorem~\ref{thm:MUL_using_SUM}).

\item In Section~\ref{Nonreversibility}, we show by explicit code construction that if a
sum-network has a $(k,n)$ fractional linear solution, then its reverse network 
also has a $(k,n)$ fractional linear coding solution over the same field
(Lemma~\ref{lem:rev_eq}). This also implies that both the networks have the same
linear coding capacity (Theorem~\ref{th:reverse}).
We show that this equivalence does not hold under non-linear network coding
(Theorem~\ref{lem:nonreversible_sum_network}).

\item In Section~\ref{Equivalence}, we use the constructions in Section~\ref{con_networks} and
similar known results~(\cite{dougherty2}) for communication networks to prove that for any
set of integer polynomials, there exists a sum-network which is scalar
linear solvable over a finite field $F$ if and only if the polynomials
have a common root in $F$ (Theorem~\ref{thm:polynomialcollection_sum_network}). For any finite or cofinite
(whose complement is finite) set of prime numbers, we present
a network so that for any positive integer $k$ and a field $F$
the network is $k$-length vector linear solvable over $F$ if
and only if the characteristic of $F$ belongs to the set (Theorems~\ref{thm:main} and \ref{thm:main_cofinite}).
This is stronger than the corresponding known results~(\cite{dougherty2})
for communication networks. 

\item In Section~\ref{sec:conseq}, using similar known results for communication networks, we
show the insufficiency of linear network codes for sum-networks
(Theorem~\ref{thm:linsufficient_sum_networks1})
and the unachievability of the network coding capacity of sum-networks
(Theorem~\ref{lem:capacity_unachievable_sum_network}).

\end{itemize}
We conclude with a discussion in Section~\ref{disc}.
\suppress{
\subsection{Organization of the paper}

{\bf To be appropriately combined with the previous subsection} 

The paper is organized as following. In Section \ref{sec:back},
we introduce the system model and notations used in this paper.
We discuss the equivalence between linear-networks and sum-networks
in Section \ref{linear function_sumnetworks}.
We present the
constructions of solvably equivalent networks in Section \ref{con_networks}.
The equivalence between a set of integer polynomials and
the sum-networks is shown in Section \ref{Equivalence}. Explicit
construction of sum-networks with given finite or cofinite characteristic
set is also presented in this section.
In Section \ref{Nonreversibility}, we present a code construction 
for the reverse network from a linear code for the original network
and address the relation between solvability and reversibility of a sum-network.
We show the
insufficiency of vector linear network coding for the sum-networks in
Section \ref{Insufficiency}. In Section \ref{Unachievability},
we show that there exists a sum-network whose network
coding capacity is unachievable over any finite alphabet.
The paper is concluded in Section \ref{disc}.
}

\section{System model}\label{sec:back}

A network is represented by a directed acyclic graph $G= (V,E)$, where
$V$ is a finite 
set of nodes and $E \subseteq V \times V$ is the set of edges in the network
. For any edge 
$e=(i,j)\in E$, the node $j$ is called the head of the edge and the node
$i$ is called the 
tail of the edge; and are denoted as $head (e)$ and $tail (e)$ respectively.
For a node $v$,   
$In(v)= \{e \in E \colon head (e) = v \}$ denotes the set of incoming edges
to $v$. Similarly, 
$Out(v)= \{e \in E \colon tail (e) = v \}$ denotes the set of outgoing edges
from $v$. A sequence of 
nodes $(v_1, v_2, \ldots, v_l)$ is called a path 
if $(v_i, v_{i+1}) \in E$ for $i=1,2,\ldots, l-1$. If $e_i = (v_i,v_{i+1}),
\,1\leq i\leq l-1$ denote the edges on the path then the path
is also denoted as the sequence of edges $e_1e_2\ldots e_{l-1}$.

Among the nodes, a set of nodes $S\subseteq V$ are sources and a set of nodes
$T\subseteq V$ are terminals. We assume that a source does not have any
incoming edge. Each source generates a set of random processes over
an {\it alphabet} $\cA$. All the source processes are assumed to be independent,
and each process is assumed to be independent and uniformly distributed
over the alphabet.
In general, each terminal in the network may want to recover  
the symbols generated at a specific set of sources or their functions. 
Each link is assumed to be an error-free and delay-free communication link
which is capable of carrying a symbol from $\cA$ in each use.

A network code is an assignment of an edge function 
for each edge and a decoding function for each terminal.
A $(k,n)$ fractional network code over an alphabet $\cA$ consists of 
an edge function for every edge $e\in E$
\begin{eqnarray}
f_{e} \colon \cA^{k}\rightarrow \cA^{n}, \mbox{if } tail (e) \in S \label{code1},
\end{eqnarray}
\begin{eqnarray}
f_{e} \colon \cA^{n|In(v)|}\rightarrow \cA^{n}, \mbox{if } tail (e)  \notin S \label{code2}
\end{eqnarray}
and a decoding function for every terminal $v$ 
\begin{eqnarray}
g_{v} \colon \cA^{n|In(v)|}\rightarrow \cA^{k}. \label{code3}
\end{eqnarray}
Such a network code fulfills the requirements of every terminal in the network
$k$ times in $n$ use of the network. The ratio $k/n$ is the rate of the code.
If $k=n$, then the code is called a $n$-length vector network code, and
if $k=n=1$, then the code is called a scalar network code.

When $\cA$ is a field (or more generally, a module over a commutative ring),
a network code is said to be linear if all the edge functions and the decoding
functions are linear over that field (respectively ring).
For any edge $e\in E$, let $Y_e $ denote the symbol transmitted through
$e$. When $\cA$ is a field, the
symbol vector carried by an edge $e$ in a $(k,n)$-fractional linear network 
code is of the form
\begin{eqnarray}
Y_e & = & 
\begin{cases}
\sum_{e^\prime : head(e^\prime ) = tail(e)} \beta_{e^\prime, e}
Y_{e^\prime}\label{code4} \,\,\text{ if } tail(e)\notin S \\
\sum_{j : X_j \mbox{ generated at }tail(e)} \alpha_{j, e}
X_{j} \, \, \text{ if } tail(e)\in S,
\end{cases}
\label{code5}
\end{eqnarray}
where $\beta_{e^\prime, e} \in \cA^{n \times n}$ for $e,e' \in E$
and $\alpha_{j,e} \in \cA^{n \times k}$ for $e\in E, j:X_j \mbox{ generated
at }tail(e)$.
A symbol vector recovered by a terminal node $v$ is given by
\begin{eqnarray}
R_{v,i} & = & \sum_{e \in In(v)} \gamma_{v,i,e} Y_e \label{code6}
\end{eqnarray}
where $\gamma_{v,i,e} \in \cA^{k \times n}$. The coefficients $\beta_{e^\prime, e},
\alpha_{j,e}$, and $\gamma_{v,i,e} $ are called the local coding coefficients.
In an $n$-length vector linear network code, these coefficients are $n\times n$
matrices, and in a scalar linear network code, they are scalars.

Given a sum-network $\cN$, its reverse network $\cN^\prime$ is defined 
as the network with the same set of vertices, the edges reversed
keeping their capacities same, and the roles of the sources and the terminals 
interchanged.

\label{page:capacity}
For a given network, a rate $r$ is said to be {\it achievable} if there
exists a $(k,n)$ fractional network code so that
$k/n \geq r$. This is similar to the definition used in \cite{dougherty4}
where the requirement is $k/n=r$. 
We note that our definition is different from the more common definition of
achievability in the information theory literature where existence
of codes of rates arbitrarily close to $r$ suffices.
The {\it coding capacity}, or simply the {\it capacity}, of a network is defined
to be the supremum of all achievable rates. Note that the capacity
is not necessarily achievable according to our definition.
The {\it linear coding capacity} is defined to be the supremum of the rates
achievable by linear network codes.


\section{Solvably equivalent networks}
\label{con_networks}

In this section, we give two constructions of equivalent networks. First  
a construction of a sum-network from a multiple-unicast network 
is presented where
the constructed network is solvably equivalent (and also linear
solvably equivalent) with the original network. It is also shown that
the reverse network of the constructed sum-network and that of the
multiple-unicast network are also equivalent.
The second construction gives a multiple-unicast network from a given
sum-network so that the constructed network and the sum-network are
linear solvably equivalent.
 
Before presenting the constructions, we note that
for any linear function of the source symbols, the linear function can
be communicated to the terminals if the sum can be communicated since
the sources themselves can pre-multiply the symbols by the respective coefficients.
The converse also holds if all the coefficients are invertible.
So the results of this paper hold true for linear functions in general.

\vspace*{2mm}
\noindent
\textbf{Construction $\textbf{C}_1$} : {\it 
A sum-network from a multiple-unicast network}:
Consider any multiple-unicast network $\net_1$ represented by the
dotted box in Fig. \ref{fig:sum_network}. $\net_1$ has $m$ sources
$w_1,w_2,\ldots, w_m$ and $m$ corresponding terminals $z_1,z_2,\ldots,z_m$
respectively.
Fig. \ref{fig:sum_network} shows a sum-network $\net_2 = C_1(\net_1)$
constructed with
$\net_1$ as a part. In this network, there are $m+1$ sources $s_1, s_2,
\ldots, s_{m+1}$ and $2m$ terminals 
$$\{t_{L_i}, t_{R_i} |1\leq i\leq m\}.$$
The reverse networks of $\net_1$ and $\net_2$ are denoted by
$\net_1^\prime$ and $\net_2^\prime$ respectively, and are shown in Fig.
\ref{fig:reverse}. The set of vertices and the set of edges of $\net_2$
are respectively 
\begin{eqnarray}
V(\net_2) &  = & V(\net_1) \cup \{s_1,s_2,\ldots, s_{m+1}\} 
 \cup \{u_1,u_2, \ldots, u_m\} \nonumber \\
&& \cup \{v_1,v_2,\ldots, v_m\} \cup \{t_{L_1},t_{L_2},\ldots, t_{L_m}\} \nonumber \\
&& \cup \{t_{R_1},t_{R_2},\ldots,t_{R_m}\} \nonumber
\end{eqnarray}
and
\begin{eqnarray}
E(\net_2) & = & E(\net_1) \cup \{(s_i,w_i)|i=1,2,\ldots, m\} \nonumber \\
&&  \cup \{(s_i, u_j)| i,j=1,2,\ldots, m, i\neq j\} \nonumber \\
&& \cup \{(s_{m+1},u_j)| j=1,2,\ldots, m\} \nonumber \\
&& \cup \{(s_i, t_{R_i})| i=1,2,\ldots, m\} \nonumber \\
&& \cup \{(u_i, v_i)| i=1,2,\ldots, m\} \nonumber \\
&& \cup \{(v_i, t_{L_i})| i=1,2,\ldots, m\} \nonumber \\
&& \cup \{(v_i, t_{R_i})| i=1,2,\ldots, m\} \nonumber \\
&& \cup \{(z_i, t_{L_i})| i=1,2,\ldots, m\}.\nonumber
\end{eqnarray}

\begin{figure*}[t]
\begin{minipage}[b]{0.48\linewidth}
\centerline{\includegraphics[width=3.2in]{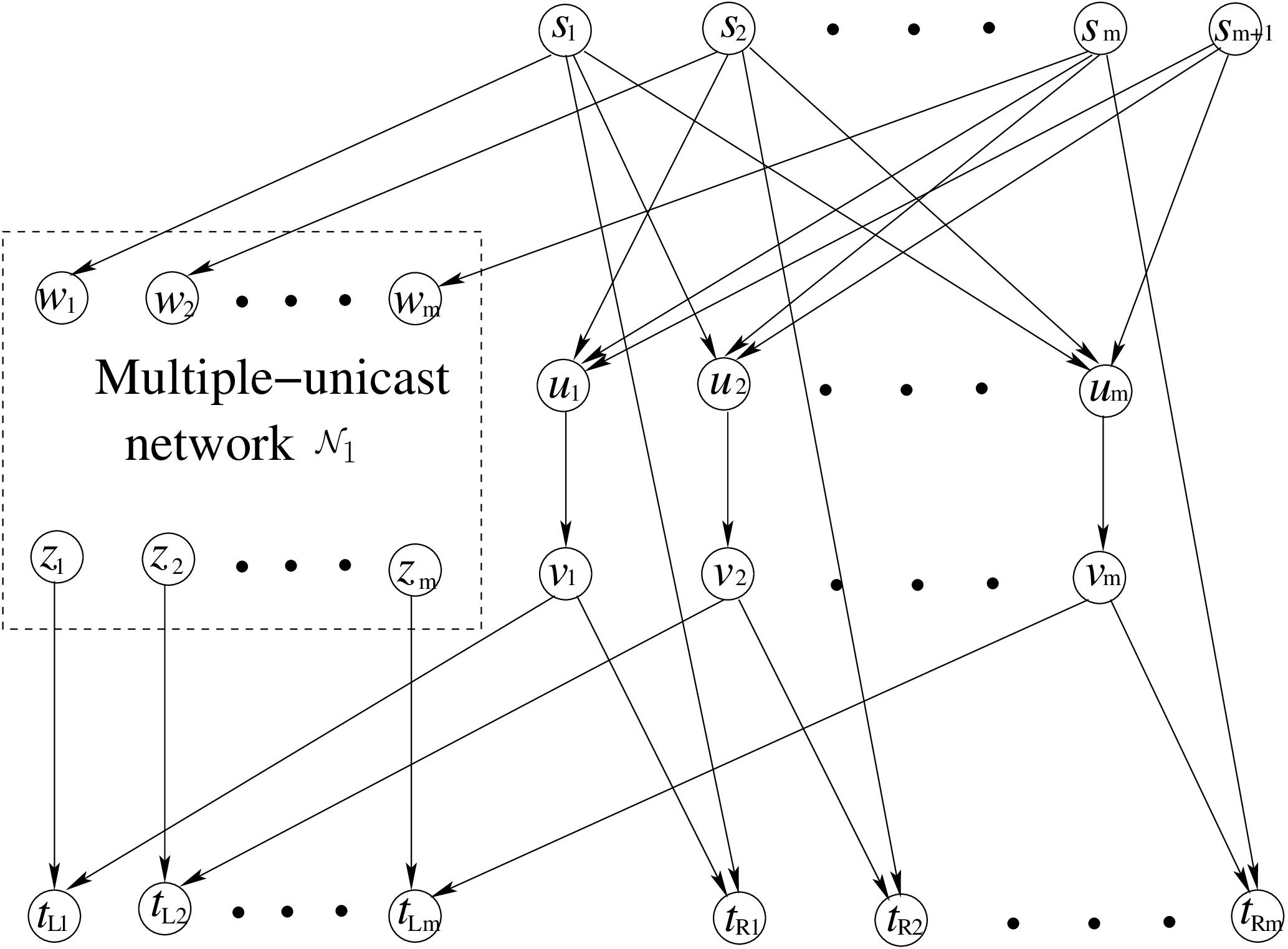}}
\caption{The construction ($C_1$) of a sum-network $\net_2$ from the multiple-unicast network $\net_1$}
\label{fig:sum_network}
\end{minipage}
\hspace{.3cm}
\begin{minipage}[b]{0.48\linewidth}
\centerline{\includegraphics[width=3.2in]{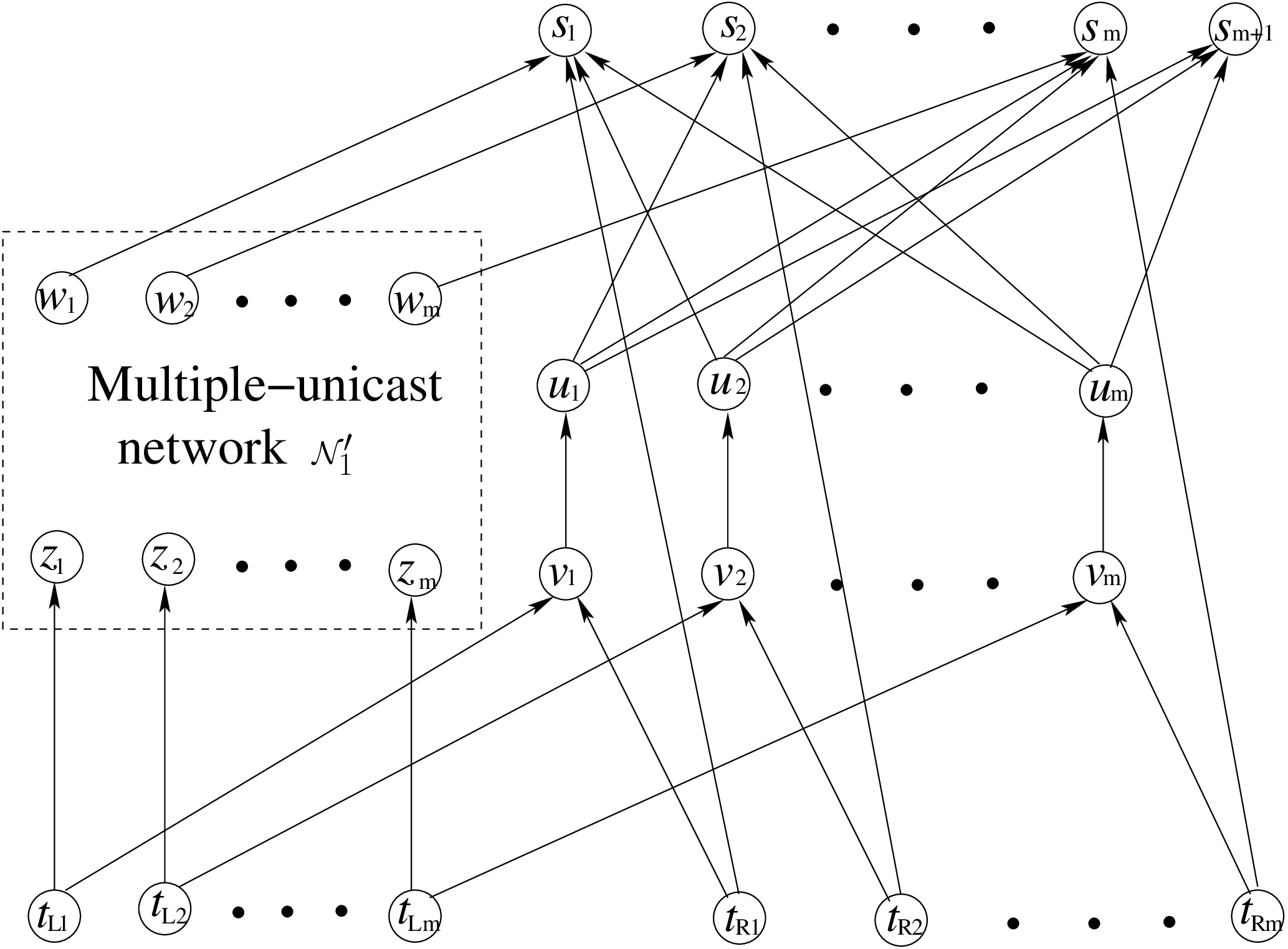}}
\caption{The reverse network $\net_2^{\prime}$ of $\net_2$}\label{fig:reverse}
\end{minipage}
\end{figure*}


The following theorem shows that the networks $\net_1$ and
$\net_2$, as well as their reverse networks, are solvably equivalent
in a very strong sense. Here $F$ denotes a field and $G$ denotes an
abelian group.

\begin{theorem}
\label{thm:SN_using_MUN}
(i) The sum-network $\net_2$ is $k$-length vector linear solvable over 
$F$ if and only if the multiple-unicast network $\net_1$ is $k$-length
vector linear solvable over $F$.\\
(ii) The sum-network $\net_2$ is solvable over $G$ if and only
if the multiple-unicast network $\net_1$ is solvable over $G$.\\
(iii) The reverse sum-network $\net_2^{\prime}$ is $k$-length vector linear solvable over 
$F$ if and only if the reverse multiple-unicast network $\net_1^{\prime}$ is $k$-length
vector linear solvable over $F$.\\
(iv) The reverse sum-network $\net_2^{\prime}$ is solvable over $G$ if and only
if the reverse multiple-unicast network $\net_1^{\prime}$ is solvable over $G$.
\end{theorem}

The proof of the theorem is given in Appendix \ref{th1proof}.

\begin{remark}
A. Though parts (i) and (iii) are stated for a finite field,
the same results can be shown to hold 
over any finite commutative ring $R$ with identity, any $R$-module with
annihilator~(\cite{atiyah}) $\{0\}$, and for 
more general forms of linear network codes defined in~\cite{jaggi2}. B. In
parts (ii) and (iv), solvability is not restricted to codes where nodes
perform only the group operation in $G$, but includes arbitrary non-linear
coding. The alphabet is restricted
to an abelian group only for defining the {\em sum}. C. For $k=1$,
the parts (i) and (iii) of the theorem gives the equivalence of the
networks for scalar solvability as a special case.
\end{remark}

It is well-known (\cite{dougherty3}) that one can construct a solvably
equivalent multiple-unicast network from any communication network.
We can then construct a solvably equivalent sum-network from this
derived multiple-unicast network using Construction $C_1$. As a result,
we get a sum-network which is solvably equivalent, and also linear
solvably equivalent, to the original communication network.

Though Construction $C_1$ gives a sum-network which is equivalent
under vector network coding solution, the equivalence does not hold
under fractional network coding. Though, as Lemma~\ref{lem:frac} below states,
a $(k,n)$ fractional solution of $\net_1$ ensures the existence of a
$(k,n)$ fractional solution of $C_1(\net_1)$, the converse does not hold.
The constructed sum-network may have a $(k,n)$
fractional solution even though the original network does not have
a $(k,n)$ fractional solution. So the coding capacity
of the constructed network may be more than that
of the original network. For instance, consider the multiple-unicast
network where $m$ source-terminal pairs
are connected through a single bottleneck link.
The multiple-unicast network has capacity $1/m$, whereas the constructed
sum-network, shown in Fig.~\ref{fig:MUN1_sum},
has a $(1,2)$ fractional vector linear solution, and thus, has a capacity
at least $1/2$. In the first
time-slot, the bottle-neck link in the multiple-unicast network
can carry the sum $X_1+X_2+\ldots + X_m$ which is then forwarded
to all the $m$ left hand side terminals of the sum-network. The links
$(u_i, v_i)$ carry only $X_{m+1}$ in the first time-slot.
So, by using one time-slot, the left hand side terminals recover the
sum $X_1+X_2+\ldots + X_{m+1}$. In the second time-slot,
the network can obviously be used to communicate the sum 
$X_1+X_2+\ldots + X_{m+1}$ to the right hand side terminals since there
is exactly one path between any source and any of these terminals.

\begin{lemma}
\label{lem:frac}
For some $k\leq n$, if a multiple-unicast network $\net$ has a $(k,n)$ fractional coding (resp. fractional linear) solution,
then the sum-network $C_1 (\net )$ also has a $(k,n)$ fractional coding (resp. fractional
linear) solution.
\end{lemma}
\begin{proof}
The proof is obvious from the construction, and is omitted.
\end{proof}

The following theorem gives a bound on the capacity of the 
constructed network $C_1(\net)$.

\begin{theorem}
\label{th:capacity}
For any multiple-unicast network $\net$,
$$
\min (1, Capacity (\net )) \leq Capacity (C_1 (\net )) \leq 1.
$$
\end{theorem}
\begin{proof}
The first inequality follows from the preceding lemma. The second
inequality follows from the fact that the min-cut from source
$s_{m+1}$ to each terminal is $1$.
\end{proof}
\begin{corollary}
\label{cor:capacity}
If a multiple-unicast network $\net$ has capacity 1, then the capacity of
the sum-network $C_1 (\net)$ is also 1.
\end{corollary}

\begin{figure*}[t]
\begin{minipage}[b]{0.7\linewidth}
\centerline{\includegraphics[width=3.5in]{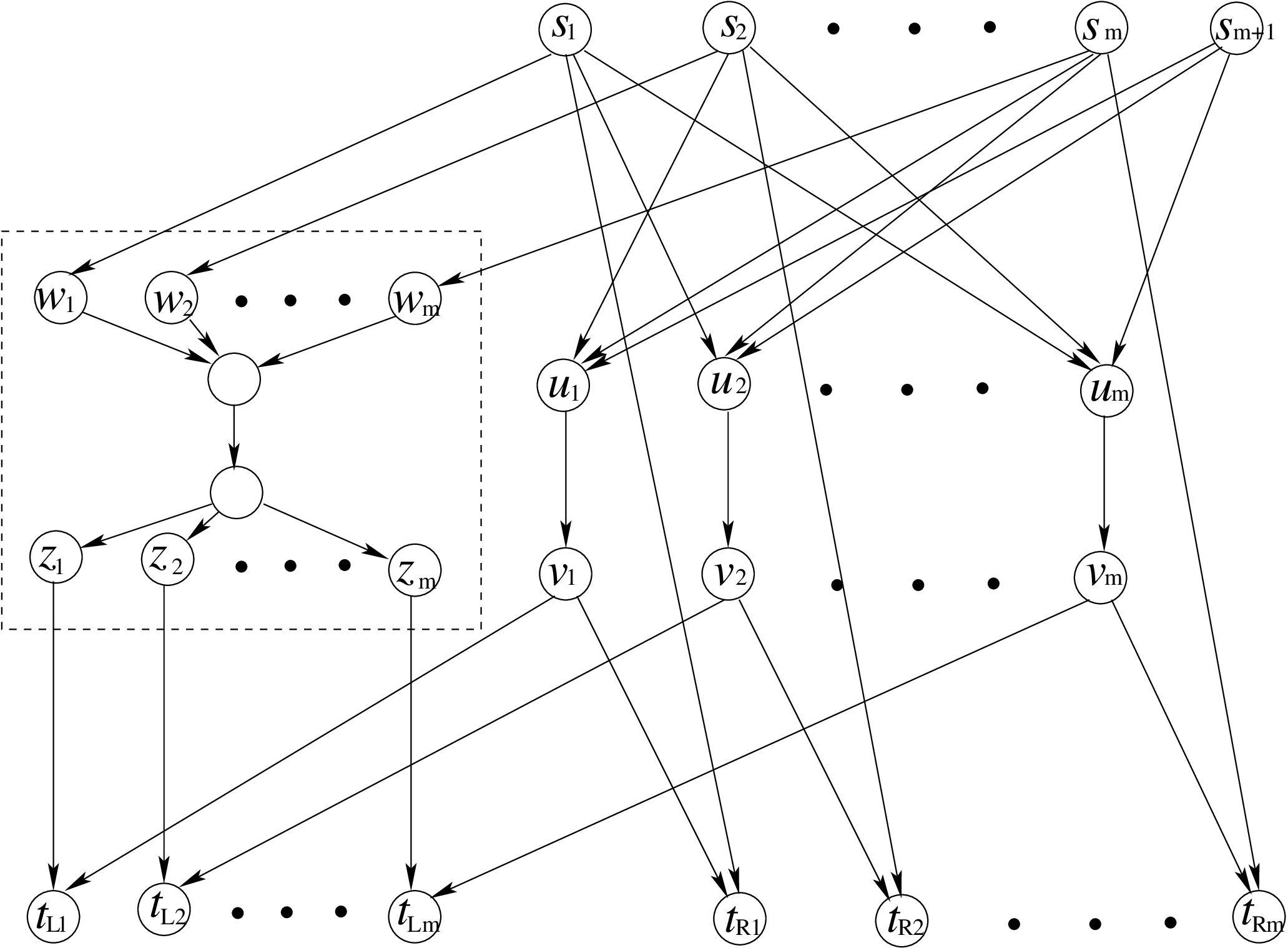}}
\caption{A Multiple-unicast network with capacity $1/m$ for which
Construction $C_1$ gives a sum-network of capacity $\geq 1/2$}
\label{fig:MUN1_sum}
\end{minipage}
\hspace{.3cm}
\begin{minipage}[b]{0.25\linewidth}
\centerline{\includegraphics[width=1.3in]{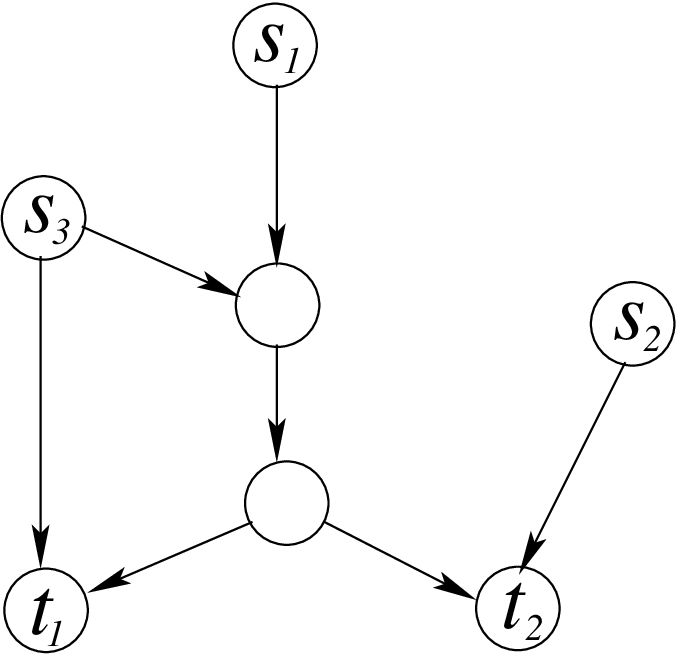}}
\caption{A component network from \cite{dougherty3} that is used
in Construction $C_2$}
\label{fig:component}
\end{minipage}
\end{figure*}

\vspace*{3mm}
\noindent
\textbf{Construction $\textbf{C}_2$ :} {\it A multiple-unicast network from a sum-network}:
Consider any sum-network $\net$ represented by the dotted box in Fig. 
\ref{fig:sum_to_multiple_unicast}. The sum-network has $m$ sources $w_1,w_2,\ldots,
w_m$ and $n$ terminals $z_1,z_2,\ldots,z_n$. The figure
shows a multiple-unicast network $C_2 (\net )$
of which $\net$ is a part. In this multiple-unicast network,
the source-terminal pairs are $\{(s_i, t_i)|i=1,2,\ldots, m\}\cup
\{(s_{ij}, t_{ij})| 1\leq i\leq m, 2\leq j\leq n\}$.

\begin{figure*}[t]
\centering
\includegraphics[width=\textwidth]{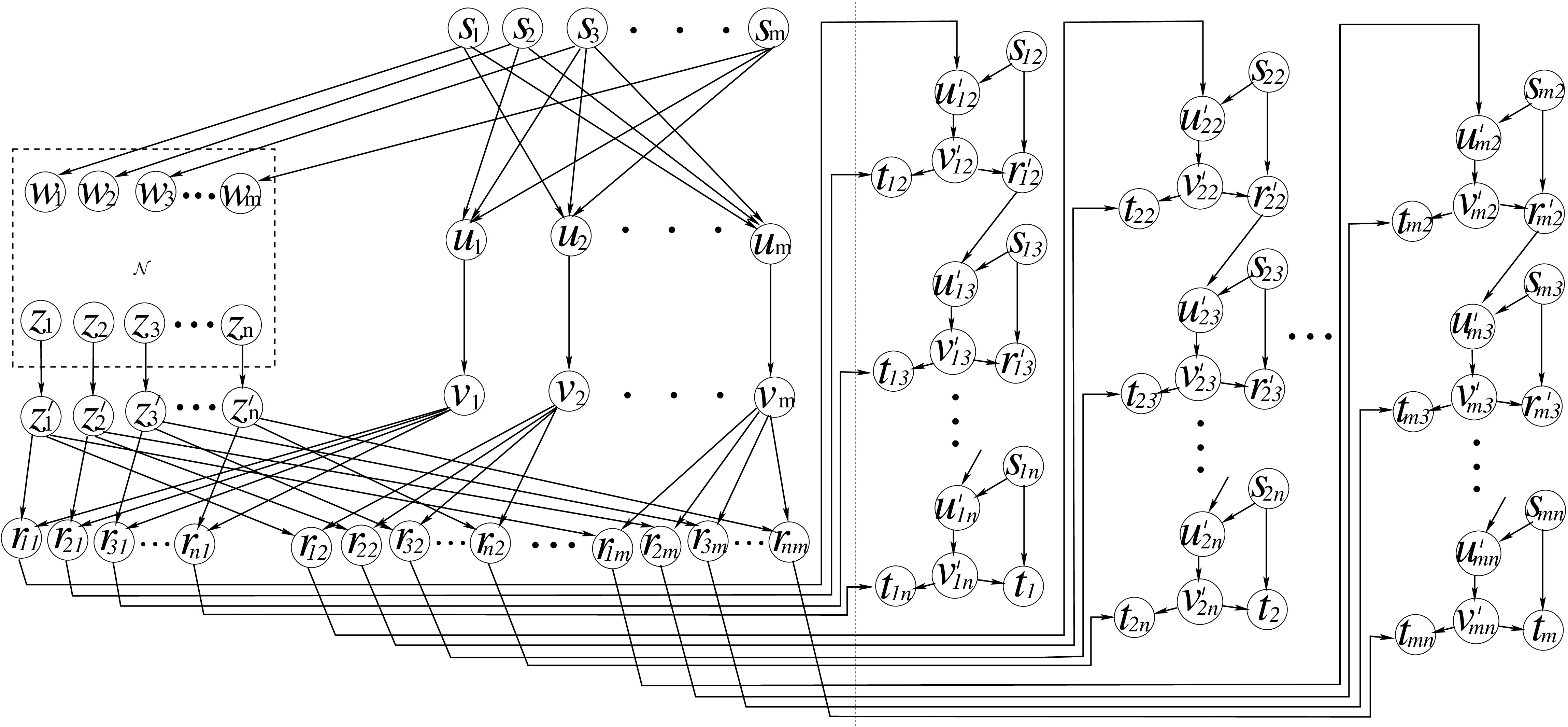}
\caption{Construction $C_2$ of a multiple-unicast network $C_2(\net)$
from a sum-network $\net$ shown in the dotted box}
\label{fig:sum_to_multiple_unicast}
\end{figure*}

\label{page:C2}
The nodes and the edges of $C_2(\net)$ are respectively
\begin{eqnarray}
V(C_2(\net)) = && V(\net ) \cup_{1\leq i\leq m} \{s_i,u_i,v_i,t_i \} \nonumber \\
&& \cup  \{z_j^\prime | 1\leq j\leq n\} \cup_{1\leq i\leq m\atop{2\leq j\leq n}} \{s_{ij},t_{ij},u^\prime_{ij},v^\prime_{ij}\} \nonumber \\
&& \cup \{r_{ji} | 1\leq i\leq m,1\leq j\leq n\} \nonumber \\
&& \cup \{r_{ij}^\prime| 1\leq i\leq m, 2\leq j\leq n-1\} \nonumber 
\end{eqnarray}
and
\begin{eqnarray}
&& \hspace*{-10mm} E(C_2(\net))  =  E(\net ) 
\cup \{(s_i,u_j)|1\leq i,j\leq m, i\neq j\} \nonumber \\
&& \cup_{1\leq i\leq m} \{(s_i,w_i),(u_i,v_i),(s_{in},t_i),(v_{in}^\prime,t_i),(r_{1i},u^\prime_{i2})\} \nonumber \\
&& \cup_{{1\leq i\leq m}\atop{1\leq j \leq n}}\{(z_j^\prime, r_{ji}),(v_i, r_{ji})\} 
    \cup \{(z_i,z_i^\prime)|1\leq i\leq n\}  \nonumber \\
 && \cup_{{1\leq i\leq m}\atop{2\leq j\leq n}}
\{(r_{ji},t_{ij}),(s_{ij},u_{ij}^\prime),(u_{ij}^\prime,v_{ij}^\prime),
     (v_{ij}^\prime,t_{ij})\} \nonumber \\
&& \cup_{{1\leq i\leq m}\atop{2\leq j \leq n-1}}\{(s_{ij},r_{ij}^\prime),(v_{ij}^\prime,r_{ij}^\prime),(r_{ij}^\prime, u_{i(j+1)}^\prime)\} . \nonumber
\end{eqnarray}
The right half of the figure is constructed using a method used in
\cite{dougherty3}. It consists of $m$ chains, each with $n-1$ copies
of the network shown in Fig.~\ref{fig:component}.
This component network in Fig.~\ref{fig:component} has the property that
if $t_2$ wants to recover the symbol generated by $s_3$, and
$t_1$ wants to recover an independent symbol $X_1$ generated by $s_1$, then
$s_1$ and $s_2$ both must send $X_1$ on the outgoing links~(\cite{dougherty3}).
The following theorem states the equivalence between $\net$ and $C_2(\net )$.

\begin{theorem}
\label{thm:MUL_using_SUM}
The multiple-unicast network $C_2( \net )$ is $k$-length vector linear solvable over 
a finite field $F$ if and only if the sum-network $\net$ 
is $k$-length vector linear solvable over $F$.
\end{theorem}

\begin{proof}
First, if the sum-network $\net$ in the dotted box is $k$-length vector linear
solvable over $F$, then it is clear that the code can be extended to
a code that solves the multiple-unicast network $C_2( \net )$.
Next, we assume that the multiple-unicast network $C_2( \net )$ is $k$-length
vector linear solvable over $F$, and prove that the sum-network
$\net$ is also $k$-length vector linear solvable over $F$. 
It can be seen by similar arguments as used in the proof
of \cite[Theorem II.1]{dougherty3} that all the terminals can recover
the respective source symbols
if and only if each of the intermediate nodes $r_{ji};j=1,2,\ldots, n,
i=1,2, \cdots, m$ can recover $X_i$.
So, for any given $k$-length vector linear solution for $C_2( \net )$,
the intermediate nodes $r_{ji}$ recover the respective $X_i$.
It implies that there is a path from each source $w_i$
to each terminal $z_j$ in the sum-network.
If $m=2$ or $n=2$, then by the results in \cite{ramamoorthy}, the sum-network is scalar linear
solvable over any field. Now let us assume $m,n \geq 3$.

\label{page:wlog}
In general, for any network and any network code over an alphabet $\cA$,
suppose a node $u$ has a single message vector $Y$, either received
on a single incoming edge or generated locally if the node is a source,
and an outgoing edge $e$ carries a function $f(Y)$. If we replace the function
by the identity function ($e$ now carries $Y$) and the function $f$
precedes any subsequent operation on the received message
at the head (receiver) node of $e$, then clearly the new network
code still serves exactly the same overall function, and the new network code
is an equivalent code. If the original code is a linear code, then the
new code is also a linear code.
Thus, without loss of generality, we assume that
\begin{eqnarray}
&& Y_{(s_i, u_j)} = X_i \mbox{ for } 1\leq i,j \leq m, i \neq j, \nonumber \\
&& Y_{(s_i, w_i)} = X_i \mbox{ for } 1\leq i \leq m, \nonumber \\
&& Y_{(v_i, r_{ji})} = Y_{(u_i, v_i)}  \mbox{ for } 1\leq i \leq m, 1\leq j \leq n, \nonumber \\
&& Y_{(z_j^{\prime}, r_{ji})} = Y_{(z_j, z_j^{\prime})} \mbox{ for } 1\leq i \leq m, 1\leq j \leq n. \nonumber
\end{eqnarray}
Let us also assume that
\begin{eqnarray}\label{eq:encode_klinear_sumtomul}
Y_{(u_i, v_i)} &=& \mathop{\sum_{j=1}^{m}}_{j\neq i} \beta_{ji} Y_{(s_j, u_i)} \mbox{ for } 1\leq i \leq m, \label{smcode10}\\
Y_{(z_i,z_i^{\prime})} &=& \mathop{\sum_{j=1}^{m}} \eta_{ij} X_j  \mbox{ for } 1\leq i \leq n, 
\end{eqnarray}
 where $\beta_{ij}, \eta_{ij} \in F^{k \times k}$.

Let us assume that the symbols recovered at the nodes $r_{ij}$ for
forwarding on the outgoing links are
\begin{eqnarray}\label{eq:decode_klinear_sumtomul}
R_{ij} & = & \gamma_{ij}^{\prime} Y_{(z_i^{\prime}, r_{ij})} + \gamma_{ij}
Y_{(v_j, r_{ij})} \nonumber \\
& = &  \gamma_{ij}^{\prime} \mathop{\sum_{l=1}^{m}} \eta_{il} X_l + \gamma_{ij}
\mathop{\sum_{l=1}^{m}}_{l\neq j} \beta_{lj} Y_{(s_l, u_j)} \nonumber \\
& = & \gamma_{ij}^{\prime} \mathop{\sum_{l=1}^{m}} \eta_{il} X_l + \gamma_{ij}
\mathop{\sum_{l=1}^{m}}_{l\neq j} \beta_{lj} X_l \nonumber\\
& = & \gamma_{ij}^{\prime} \eta_{ij} X_j + \mathop{\sum_{l=1}^{m}}_{l\neq j}
(\gamma_{ij}^{\prime}\eta_{il}+\gamma_{ij}\beta_{lj}) X_l \nonumber
\end{eqnarray}
for $1\leq i \leq n, 1\leq j \leq m$.
Since the node $r_{ij}$ recovers $X_j$, we have
\begin{subequations}\label{eq:code2_klinear_sumtomul}
\begin{eqnarray}
&& \gamma_{ij}^{\prime} \eta_{ij} = I   \mbox{ for } 1\leq i \leq n, 1\leq j \leq m, \label{code17_sumtomul}\\
&& \gamma_{ij}^{\prime}\eta_{il}+\gamma_{ij}\beta_{lj} = 0 \nonumber \mbox{ for } 1\leq i \leq n, 
\,\, 1\leq l,j \leq m,\,\, l \neq j. \label{code18_sumtomul}
\end{eqnarray}
\end{subequations}
It follows from (\ref{code17_sumtomul}) that the matrices
$ \gamma_{ij}^{\prime}$ and $\eta_{ij}$ are invertible for all $i,j$.
Then it also follows from
(\ref{code18_sumtomul}) that the matrices $\gamma_{ij}$ and $\beta_{lj}$ are
also invertible for all $i,j,l$ in their range with $l\neq j$.

We will now prove that $Y_{(z_i, z_i^\prime )}$ for different $i$
are scaled versions of each other. That is, the terminals of the
sum-network recover essentially the same linear combination of
the sources. For this, we need to prove that for any $l,l^\prime$,
$\eta_{il}^{-1}\eta_{il^\prime }$ is independent of $i$.
Let us take a $j\neq l,l^\prime$. This is possible since $m>2$.
Eq. (\ref{code18_sumtomul}) gives 
\begin{eqnarray}
\eta_{il} & = & - \gamma_{ij}^{\prime -1} \gamma_{ij}\beta_{lj}, \nonumber \\
\eta_{il^\prime } & = & - \gamma_{ij}^{\prime -1} \gamma_{ij}\beta_{l^\prime j}. \nonumber 
\end{eqnarray}
These equations give
\begin{eqnarray}
\eta_{il}^{-1}\eta_{il^\prime } = \beta_{lj}^{-1}\beta_{l^\prime j}, \nonumber 
\end{eqnarray}
and so this is independent of $i$. This proves that it is possible
to communicate a fixed linear combination of the sources through
the sum-network $\net$, where each linear coefficient matrix is
invertible. If the sources themselves pre-multiply the source
messages by the inverse of the respective linear coefficient matrix, then the
terminals can recover the sum of the sources. This completes the proof.
\end{proof}


\section{Codes for the reverse of a sum-network}
\label{Nonreversibility}

Recall that, given a sum-network $\cN$, its reverse network $\cN^\prime$ is defined 
to be the network with the same set of vertices, the edges reversed
keeping their capacities same, and the role of sources and terminals 
interchanged.
We note that, since $\cN$ may have unequal 
number of sources and terminals, the number of sources 
(resp. terminals) in $\cN$ and that in $\cN^\prime$ may be different. 
We will prove in this section that a sum-network has a $(k,n)$
fractional linear solution if and only if its reverse network has
a $(k,n)$ fractional linear solution. First, in the following,
using a given network code for an arbitrary network, we construct a simple 
network code
for the reverse network and investigate the properties of this new code. 

We will represent a process generated at a source by an incoming
edge at the source, and a process recovered at a terminal by an
edge outgoing from the terminal. For any edge $e$, let us denote
the corresponding edge in the opposite direction in $\net^\prime$
by $\tilde{e}$. If $e$ is a source process, then in $\net^\prime$,
$\tilde{e}$ denotes a recovered process at that terminal, and vice versa.
Consider any $(k,n)$ fractional linear code $\cC$ over a field $F$ for $\net$. 
Let the local coding coefficient for any two adjuscent edges $e,e^\prime$
be denoted by $\alpha_{e,e^\prime}$. Note that the local
coding coefficient for any two adjuscent internal edges is a $n\times n$ matrix;
for a source node, the local coding coefficient between a source process
and an outgoing edge is a $n\times k$ matrix; and for a terminal node,
the local coding coefficient between an incoming edge and a recovered
process is a $k\times n$ matrix. Let us consider a path
$P = e_1e_2\ldots e_t$ in $\net$, and the corresponding reverse path
$\widetilde{P} = \tilde{e}_t\ldots \tilde{e}_2\tilde{e}_1$ in $\net^\prime$.
Here $e_1$ may denote a source process
and $e_t$ may denote a recovered process at a terminal.
The product $\alpha_{e_{t-1},e_t}\cdots \alpha_{e_2,e_3}\alpha_{e_1,e_2}$
is the gain $G_P$ of this path. $G_P$ is a matrix of dimension
$n\times n$ if both $e_1$ and $e_t$ are internal edges of the network,
$k\times n$ if $e_1$ is an internal edge and $e_t$ is a recovered process
at a terminal, $n\times k$ if $e_1$ is a source process and $e_t$ is
an internal edge, and $k\times k$ if $e_1$ is a source process and $e_t$ is
a recovered process at a terminal.

Consider the code $\cC^\prime$ for $\net^\prime$ given by
the local coding coefficients $\beta_{\tilde{e}^\prime , \tilde{e}}
= \alpha_{e,e^\prime}^T$, where `$T$' denotes transpose. 
We call this code the {\em canonical reverse code} of $\cC$.
Under this code, the path gain of $\widetilde{P}$ is 
\begin{eqnarray}
&& G_{\widetilde{P}} = \alpha_{e_1,e_2}^T  \alpha_{e_2,e_3}^T \cdots
\alpha_{e_{t-1},e_t}^T = G_P^T. \label{eq:pgain}
\end{eqnarray}
This code can be shown to be the same as the {\em dual code} defined
by Koetter et al. in \cite{koetter3}. They used this code to show
the equivalence of a network and its reverse entwork under linear solvability
for multiple-unicast networks and multicast-networks. Though
their suggested application for the reverse-multicast network
was the centralized detection of an event sensed by exactly one
of the many sensors deployed in an area, this is obviously a special
application of the resulting sum-network. As explained below, this code
also solves the reverse sum-network if the original code provides a solution
to the original sum-network.

Consider two cuts $\chi_1$ and $\chi_2$ in $\net$. The first
cut may include edges to the sources which correspond to the source
processes, and the second cut may include edges out of the terminals
corresponding to the recovered processes. Let the edges in $\chi_1$ and
$\chi_2$ be $e_{11},e_{12},\ldots, e_{1r}$ and $e_{21},e_{22},\ldots, e_{2l}$
respectively. The transfer matrix $T_{\chi_1 \chi_2}$ between the
two cuts relates the messeges carried by the two cuts as 
\begin{eqnarray}
\left[\begin{array}{c} Y_{e_{21}}\\ Y_{e_{22}} \\ \vdots \\ Y_{e_{2l}}
\end{array}\right] = 
T_{\chi_1 \chi_2} \left[\begin{array}{c} Y_{e_{11}}\\ Y_{e_{12}} \\ \vdots \\ Y_{e_{1r}}
\end{array}\right].
\end{eqnarray}
Note that for each $i,j$, $Y_{e_{ij}}$ itself is a column vector of
dimension $n$ or $k$ depending on whether it is an internal edge or
an edge for a source process or a process recovered at a terminal.
The transfer matrix has appropriate dimension depending on the dimensions
of the column vectors on the left hand side and right hand side. It is convenient to view
$T_{\chi_1 \chi_2}$ as a $l\times r$ matrix of blocks of appropriate
sizes. The $(i,j)$-th block element of the matrix has dimension
$dim(Y_{e_{2i}})\times dim(Y_{e_{1j}})$. The $(i,j)$-th block is
the sum of the gains of all paths in the network from $e_{1j}$ to
$e_{2i}$. By (\ref{eq:pgain}), each path gain is transposed in the
reverse network under the code $\cC^\prime$. So, the transfer matrix
from the cut $\chi_2$ to the cut $\chi_1$ of $\net^\prime$ under the code
$\cC^\prime$ is given by
\begin{eqnarray}
&& \widetilde{T}_{\chi_2\chi_1} = T_{\chi_1 \chi_2}^T. \label{eq:trmat}
\end{eqnarray}
Now consider a generic sum-network $\cN$ depicted in Fig. \ref{fig:sumnet}.
Consider the {\it source-cut} $\chi_s$ and the {\it terminal-cut}
$\chi_t$ shown in the figure.
The transfer matrix from $\chi_s$ to $\chi_t$ 
is an $l\times m$  block matrix $T_{\chi_s\chi_t}$ with each block
of size $k\times k$ over $F$. It relates the vectors $\bX = (X_1^T, X_2^T, \ldots, X_m^T)^T$ and $\bR = (R_1^T, R_2^T, \ldots, R_l^T)^T$ as
$\bR = T_{\chi_s\chi_t}\bX$. Here $X_i, R_j$ are all column vectors of
length $k$.
The $(i,j)$-th element (`block' for vector linear coding) of the
transfer matrix is the sum of the path gains of all paths from
$X_j$ to $R_i$.
A $(k,n)$ network code provides a rate $k/n$ solution for the sum-network
if and only if this transfer matrix is the all-identity matrix, i.e., if
\begin{eqnarray}
&& T_{\chi_s\chi_t} = \left[\begin{array}{cccc}
I_k & I_k & \cdots & I_k \\
I_k & I_k & \cdots & I_k \\
\vdots & \vdots & \ddots & \vdots \\
I_k & I_k & \cdots & I_k \end{array}\right]_{lk\times mk}.
\label{eq:sum_trmat}
\end{eqnarray}
Clearly transposition of this matrix preserves the same
structure.
For a multiple-unicast network on the other hand, 
the transfer matrix between the source-cut and the terminal-cut 
given by a $(k,n)$ fractional linear code is
\begin{eqnarray}
&& T_{\chi_s\chi_t} = \left[\begin{array}{cccc}
I_k & {\bf 0}_k & \cdots & {\bf 0}_k \\
{\bf 0}_k & I_k & \cdots & {\bf 0}_k \\
\vdots & \vdots & \ddots & \vdots \\
{\bf 0}_k & {\bf 0}_k & \cdots & I_k \end{array}\right]_{mk\times mk},
\label{eq:mult_trmat}
\end{eqnarray}
where ${\bf 0}_k$ denotes the $k\times k$ all-zero matrix.
This matrix is symmetric, and so is invariant under transposition.
Since the desired transfer matrix for a $(k,n)$ fractional linear code
for a sum-network $\net$ and its reverse network $\net^\prime$ are the
transpose of each other, our canonical reverse code construction gives
the following lemma by (\ref{eq:trmat}).

\begin{lemma}
\label{lem:rev_eq}
A sum-network $\cN$ has a $(k,n)$ fractional linear solution
if and only if the reverse network $\cN^\prime$ also has a
$(k,n)$ fractional linear solution.
\end{lemma}
The following theorem follows directly from Lemma~\ref{lem:rev_eq}.
\begin{theorem}
\label{th:reverse}
A sum-network and its reverse network have the same linear
coding capacity.
\end{theorem}

These results also hold for any linear
function over a commutative ring as long as the linear coefficients of
the function are invertible.
Lemma \ref{lem:rev_eq} shows in particular that a sum-network and its
reverse network are solvably equivalent under fractional linear coding.
For the special case of $\min\{m,l\} \leq 2$, this was also shown in
 \cite{ramamoorthy}.
The same result also follows for multiple-unicast network from
(\ref{eq:trmat}) and (\ref{eq:mult_trmat}), and this was proved
in \cite{koetter3, riis1}.

If non-linear coding is allowed, then a multiple-unicast network
was given in \cite{dougherty3} which is solvable though
its reverse multiple-unicast network is not solvable over any finite alphabet.
Now consider the sum-network obtained by using Construction $C_1$ on
this network. By Theorem \ref{thm:SN_using_MUN} (ii) and (iv),
it follows that this network allows a non-linear coding
solution whereas its reverse sum-network does not have a solution
over any abelian group.
So, we have,

\begin{theorem}
\label{lem:nonreversible_sum_network}
There exists a solvable sum-network whose reverse network is not solvable
over any finite alphabet.
\end{theorem}


\section{Systems of polynomial equations and sum-networks}
\label{Equivalence}
It was shown in \cite{dougherty2} that for any collection of 
polynomials having integer coefficients, there exists a directed acyclic 
network, and thus also a multiple-unicast network, which is scalar
linear solvable over $F$ if and only if the polynomials have a common
root in $F$. By using Construction $C_1$, it then follows that
the same also holds for sum-networks. Thus we have

\begin{theorem}
\label{thm:polynomialcollection_sum_network}
For any system of integer polynomial equations, there exists a sum-network
which is scalar linear solvable over a finite field $F$ if and only if the
system of polynomial equations has a solution in $F$.
\end{theorem}

For a specific class of networks, for example multicast networks, there
may not exist a network corresponding to any system of polynomial equations.
For example, for the class of multicast networks, there is no network
which is solvably equivalent to the polynomial equation $2X=1$. This is
because, the polynomial equation has a solution only over fields of
characteristic not equal to $2$. Whereas, if a multicast network is
solvable over any field, then it is also solvable over large enough fields
of characteristic $2$. Theorem~\ref{thm:polynomialcollection_sum_network}
thus affirms the broadness of sum-networks as a class.


\subsection{Networks with finite and cofinite characteristic sets}
\label{sec:finitecofinite}
The equivalence in Theorem~\ref{thm:polynomialcollection_sum_network}, and the corresponding original result on
communication networks hold under scalar network coding. In the following,
we present two classes of networks which are vector linear solvable for
any vector length if and only if the characteristic of the alphabet field
belongs to a given set of primes. Equivalently, this gives sum-networks
which are equivalent to special classes of polynomials under vector linear
network coding of any dimension. Thus for this classes of polynomials,
a stronger equivalence holds with these networks.

A constant polynomial $P(x)=p_1p_2\ldots p_l$ where $p_1, p_2, \ldots,
p_l$ are some prime numbers, has a solution over a field $F$ if and only
if the characteristic of $F$ is one of $p_1, p_2, \ldots, p_l$.
Theorem~\ref{thm:polynomialcollection_sum_network}
ensures the existence of a sum-network which is scalar linear solvable
only over such fields. Such a network can be constructed using the
algorithm given in \cite{dougherty2} and Construction $C_1$.
In Fig. \ref{fig:spcm}, we show a much simpler and directly constructed network
which satisfies the above property under vector linear coding of any
dimension. Here $m=p_1 p_2 \ldots p_l +2$.

\begin{theorem}[Finite characteristic set]
\label{thm:main}
For any finite, possibly empty, set $\cP = \{p_1, p_2, \ldots, p_l\}$
of positive prime numbers, there exists a directed acyclic network
of unit-capacity edges so that for any positive integer $n$, 
the network is $n$-length vector linear solvable 
if and only if the characteristic of the alphabet field belongs to $\cP$.
\end{theorem}
The proof is given in Appendix \ref{app1}.

Taking the empty set as $\cP$, we get the network $\spc_3$ shown in
Fig. \ref{fig:spc3}.
This is not solvable over any alphabet field for any vector dimension.
This network was also found independently by Ramamoorthy and 
Langberg~(\cite{ramamoorthy1}).

\begin{figure*}[t]
\begin{minipage}[b]{0.6\linewidth}
\centerline{\includegraphics[width=3.0in]{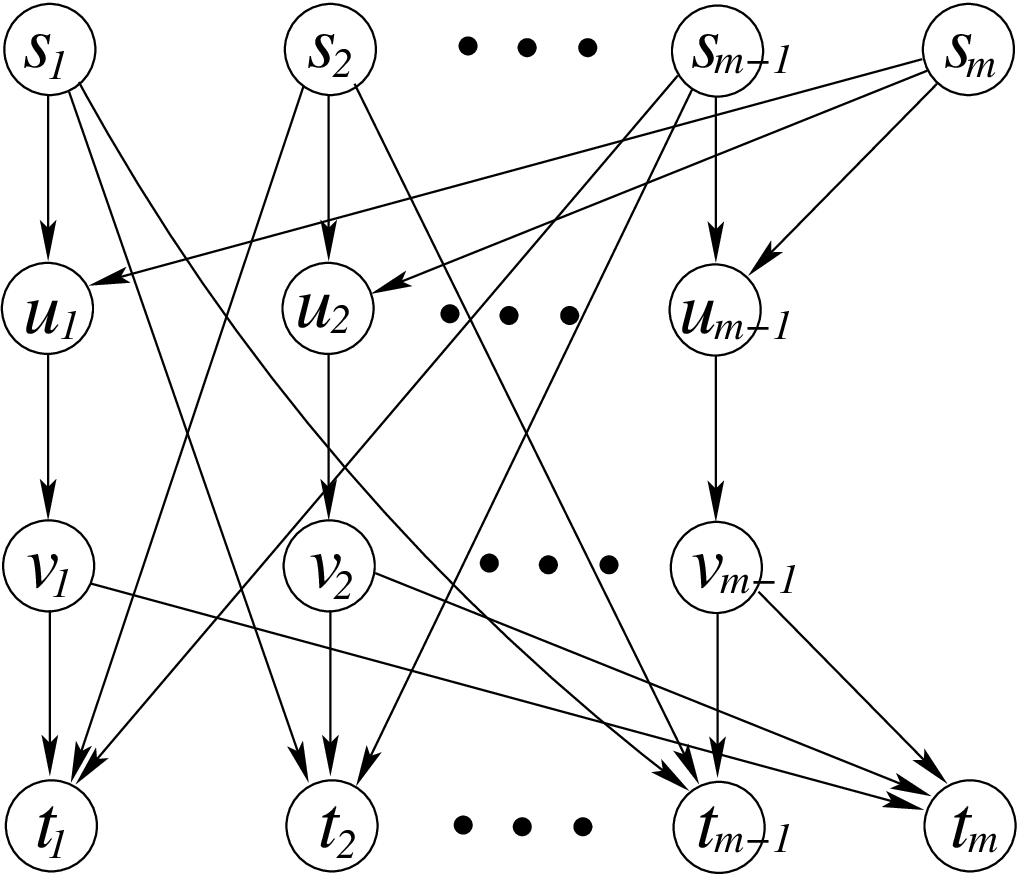}}
\caption{The network $\spc_m$}
\label{fig:spcm}
\end{minipage}
\hspace{.3cm}
\begin{minipage}[b]{0.32\linewidth}
\centerline{\includegraphics[width=1.8in]{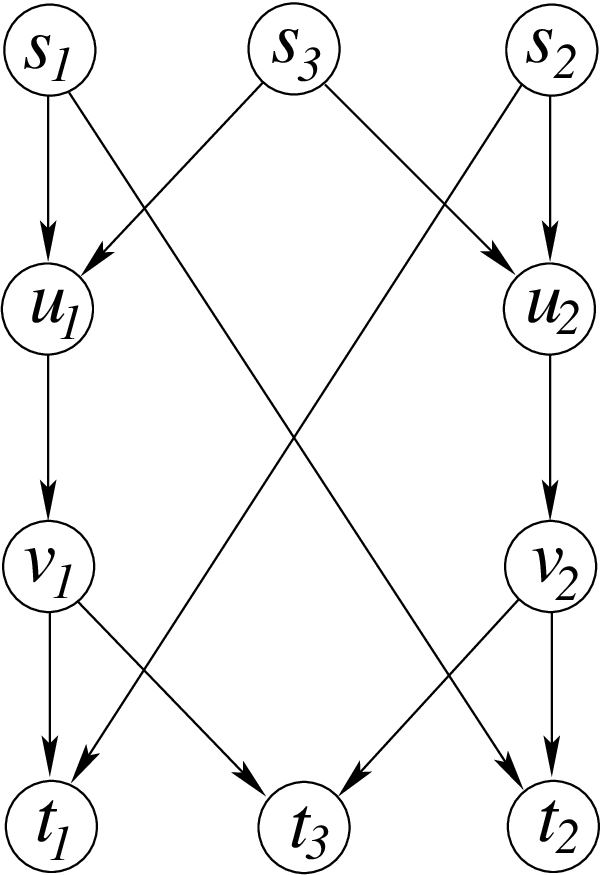}}
\caption{A sum-network $\spc_3$ which is not solvable over any field}
\label{fig:spc3}
\end{minipage}
\end{figure*}


\begin{figure*}[t]
\begin{minipage}[b]{0.6\linewidth}
\centerline{\includegraphics[width=3.0in]{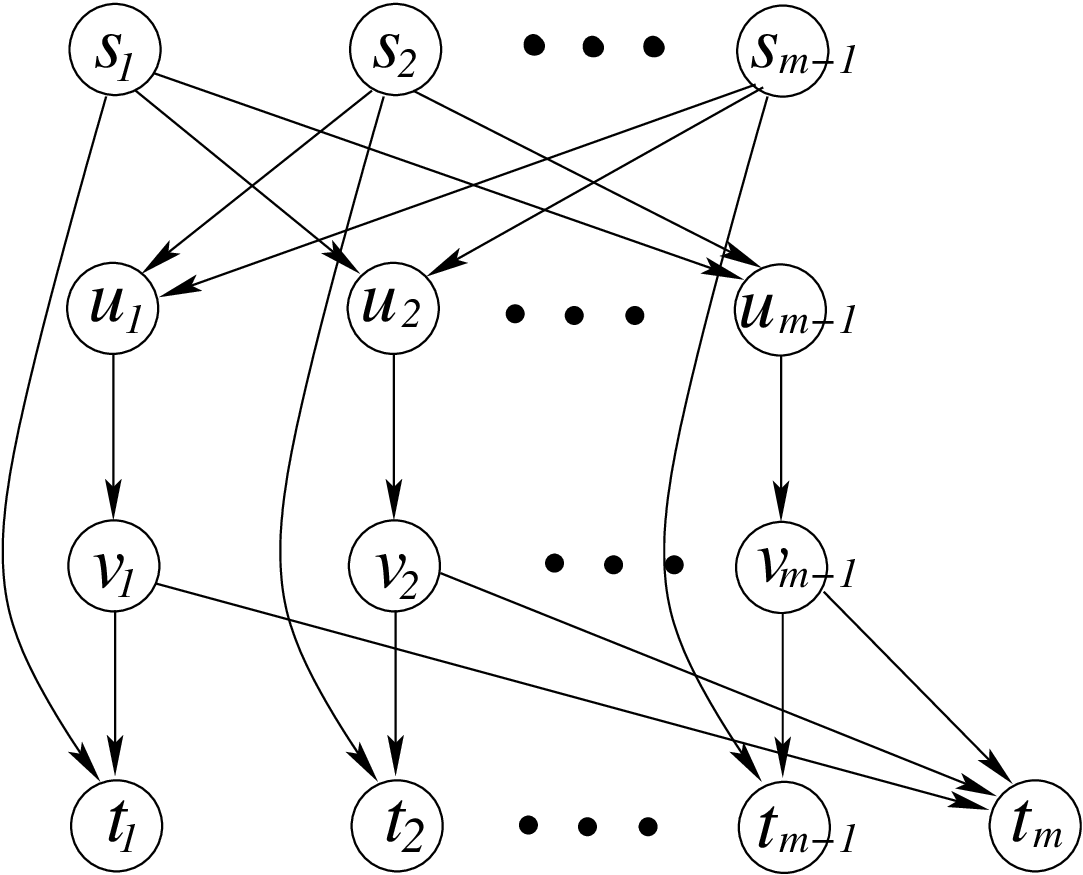}}
\caption{The network $\spc_m^*$}
\label{fig:spcm*}
\end{minipage}
\hspace{.3cm}
\begin{minipage}[b]{0.35\linewidth}
\centerline{\includegraphics[width=2.5in]{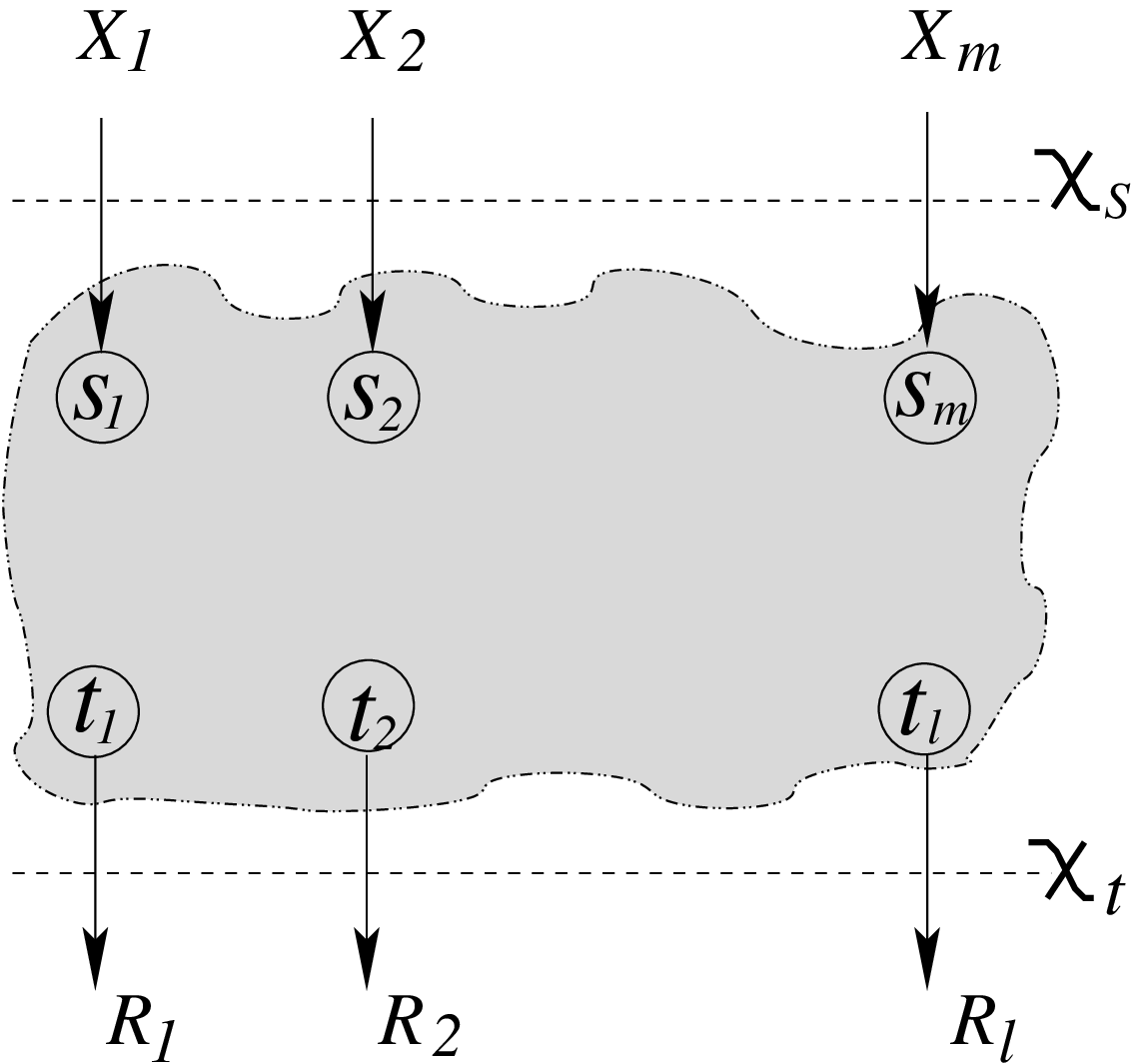}} \caption{The
source-cut and the terminal-cut of a
generic sum-network } \label{fig:sumnet}
\end{minipage}
\end{figure*}

The polynomial
$P(x)=(p_1p_2\ldots p_l)x-1$, where $p_1, p_2, \ldots, p_l$ are some
prime numbers, has a solution only over fields of characteristic not
equal to any of $p_1, p_2, \ldots, p_l$.
In Fig. \ref{fig:spcm*}, we show a network which, for any positive integer $n$,
 has a $n$-length vector linear solution over exactly those fields.
\begin{theorem}[Cofinite characteristic set]
\label{thm:main_cofinite}
For any finite set $\cP = \{p_1, p_2, \ldots, p_l\}$
of positive prime numbers, there exists a directed acyclic sum-network
of unit-capacity edges so that for any positive integer $n$, the network
is $n$-length vector linear solvable if and
only if the characteristic of the alphabet field does not belong to $\cP$.
\end{theorem}
The proof is given in Appendix \ref{app2}.

The polynomials considered in this subsection are very special in nature.
Their solvability over a field $F$ depends only on the characteristic
of $F$. Only for such a system of polynomials, it is possible to find
a network which is solvably equivalent under vector linear coding
of any dimension. This is because, otherwise if the polynomials have
a solution over a finite extension $K$ of $F$ even though they
do not have any solution over $F$, then a scalar linear solvably
equivalent network has a scalar solution over $K$, but not
over $F$. This implies that the network also has a vector
solution over $F$ of dimension $\dim_F K$ even though
it does not have a scalar solution over $F$.

We note that Theorems~\ref{thm:main} and \ref{thm:main_cofinite}
are stronger than corresponding known results~(\cite{dougherty2}) for communication
networks.
However, by applying Construction $C_2$ on the sum-networks $\spc_m$ and $\spc_m^\star$,
we can also get communication networks which satisfy the properties in
Theorems~\ref{thm:main} and \ref{thm:main_cofinite} respectively.

\section{Some consequences of the equivalence results}
\label{sec:conseq}
\subsection{Insufficiency of linear network coding for sum-networks}
\label{Insufficiency}

It was shown in \cite{dougherty1} that linear network coding may not
be sufficient for communication networks in the sense that a rate may be
achievable by non-linear coding even though the same rate may not be
achievable by linear coding over any field and of any vector dimension.
Such a communication network was presented in \cite{dougherty1}.
It was proved that for this network, there is a non-linear solution
over the ternary alphabet $F_3$ even though
there is no linear solution over any finite field for any vector
dimension. We can get an equivalent multiple-unicast network 
which also satisfies the same property.
Then the sum-network obtained using Construction $C_1$ on this multiple-unicast
network satisfies the same property by Theorem~\ref{thm:SN_using_MUN}.
So for the resulting sum-network, rate $1$ is achievable over $F_3$ using
non-linear coding but not using vector linear coding of any dimension.
So, we have,


\begin{theorem}
\label{thm:linsufficient_sum_networks1}
There exists a sum-network with the sum defined over a finite field
which is solvable by non-linear network coding, but not solvable
using vector linear coding for any vector dimension.
\end{theorem}

\begin{remark}
It was also shown in \cite{dougherty1} that the constructed
communication network is not solvable using linear network coding
over any finite commutative ring $R$ with identity, $R$-module, or
by using more general forms of the linear network codes defined
in~\cite{jaggi2}. The same results hold for the resulting sum-network
with the additional constraint that the $R$-module needs to have
the annihilator $\{0\}$ for Theorem~\ref{thm:SN_using_MUN}(i) to hold.
\end{remark}

\subsection{Unachievability of network coding capacity of sum-networks}
\label{Unachievability}

It is known that
the network coding capacity of a communication network is independent of
the alphabet~(\cite{cannons1}). Given that the majority of the network coding
literature considers zero-error recovery and noiseless link models, it is
of interest
to know if there always exists a code achieving the capacity exactly, unlike
achieving a rate close to the capacity. In
\cite{dougherty4}, a communication network (let us call it $\net_1$) was
presented for which the network coding capacity is $1$ but there is no
code of rate $\geq 1$ over any finite alphabet. 
We now argue that there also exists a sum-network whose 
network coding capacity is not achievable over any abelian group.
Consider the sum-network $\net_2$ obtained by first taking the 
equivalent multiple-unicast network of $\net_1$, and then using
Construction $C_1$ on it. 
By Theorem \ref{thm:SN_using_MUN}, rate $1$ is not achievable in
$\net_2$. But it follows from Corollary~\ref{cor:capacity} that the capacity of
$\net_2$ is $1$.
This gives the following theorem.

\begin{theorem}
\label{lem:capacity_unachievable_sum_network}
There exists a sum-network whose network coding capacity is not 
achievable.\footnote{Recall that in our definition of achievability,
a rate is achievable if there is a code achieving exactly
that rate or higher.}
\end{theorem}

\section{Discussion}
\label{disc}
The primary message of this paper is that sum-networks form
as broad a class of networks as communication networks. Explicit
construction of sum-networks from a given multiple-unicast
network, and vice-versa,  preserves the solvability of the
networks. A sum-network and its reverse network are also
shown to be linear solvably equivalent under fractional linear network coding.
The solvably equivalent
constructions of various types of networks from other networks enabled
proving similar results, e.g., equivalence with polynomial systems,
unachievability of capacity, and insufficiency of linear coding
about sum-networks by using their known counterparts for communication networks.
These equivalent constructions also prove the difficulty of designing network
codes for sum-networks from similar results on communication networks.

\section{Acknowledgment}
The authors thank Kenneth Zeger for
his feedback on an earlier version of this work and particularly for
pointing out other works of similar flavor.
The authors would like to thank Muriel M\'{e}dard for bringing
\cite{koetter3} to their attention and pointing out a possible
connection with the code-construction for the reverse network
in Section \ref{Nonreversibility}.
The authors are grateful to the Associate Editor and the anonymous
referees for their
constructive comments which helped improve the presentation of the
paper significantly.
This work was supported in part by Tata Teleservices 
IIT Bombay Center of Excellence in Telecomm (TICET). The work of
B.~K.~Dey was also supported in part by a fund from the Department of
Science and Technology, Government of India.

\bibliographystyle{ieeetr}

\appendices
\section{Proof of Theorem \ref{thm:SN_using_MUN}}
\label{th1proof}
\textbf{Proof of part $(i)$: }\ First, if the 
multiple-unicast network $\net_1$ is $\bl$-length 
vector linear solvable over $F$ then such a code can be extended to a
$n$-length vector linear code of $\net_2$ by assigning the
$n\times n$ identity matrix to all the other local
coding matrices. Clearly this gives a required solution
for $\net_2$.

Now we prove the converse. We denote the $\bl$-length symbol vector
carried by an edge $e$ by $Y_e$ as in $(\ref{code4})$. 
By the same arguments as in the proof of Theorem~\ref{thm:MUL_using_SUM} in page~\pageref{page:wlog}, without loss of generality, we assume that
\begin{eqnarray}
&& Y_{(v_i, t_{L_i})} = Y_{(v_i, t_{R_i})} = Y_{(u_i, v_i)} \mbox{ for } 1\leq i \leq m, \nonumber \\
&& Y_{(s_i, t_{R_i})} = Y_{(s_i, w_i)} = X_i \mbox{ for } 1\leq i \leq m, \nonumber \\
\mbox{ and } && Y_{(s_i, u_j)} = X_i \text{ for } 1 \leq i \leq m+1, 1 \leq j \leq m, i\neq j. \nonumber
\end{eqnarray}

Let us also assume that
\begin{subequations}\label{eq:encode_klinear}
\begin{eqnarray}
Y_{(u_i, v_i)} &=& \mathop{\sum_{j=1}^{m+1}}_{j\neq i} \beta_{j,i} Y_{(s_j, u_i)} \mbox{ for } 1\leq i \leq m, \nonumber \\
Y_{(z_i, t_{L_i})} &=& \mathop{\sum_{j=1}^{m}}
\beta_j^{i} X_j \mbox{ for } 1\leq i \leq m, \nonumber
\end{eqnarray}
\end{subequations}
and the decoded symbols at the terminals $t_{L_i}$ and $t_{R_i}$ are
respectively
\begin{subequations}\label{eq:decode_klinear}
\begin{eqnarray}
& R_{L_i}  = & \gamma_{i,L} Y_{(z_i, t_{L_i})}
+ \gamma_{i,L}^{\prime} Y_{(v_i, t_{L_i})} 
 \mbox{ for } 1\leq i \leq m, \nonumber \\
\text{ and } & R_{R_i}  = & \gamma_{i,R} Y_{(s_i, t_{R_i})}
+ \gamma_{i,R}^{\prime} Y_{(v_i, t_{R_i})} 
 \mbox{ for } 1\leq i \leq m. \nonumber 
\end{eqnarray}
\end{subequations}
Here all the coding coefficients are $\bl \times \bl$ matrices over $F$. 
From $(\ref{eq:encode_klinear})$ and $(\ref{eq:decode_klinear})$, we have 
\begin{subequations}\label{eq:decode2_klinear}
\begin{eqnarray}
R_{L_i} & = & \gamma_{i,L} \mathop{\sum_{j=1}^{m}}\beta_j^{i} X_j
+ \gamma_{i,L}^{\prime} \mathop{\sum_{j=1}^{m+1}}_{j\neq i} \beta_{j, i}X_j 
 \mbox{ for } 1\leq i \leq m, \nonumber \\
R_{R_i} & = & \gamma_{i,R} X_i + \gamma_{i,R}^{\prime} \mathop{\sum_{j=1}^{m+1}}_{j\neq i} \beta_{j,i}X_j 
 \mbox{ for } 1\leq i \leq m. \nonumber 
\end{eqnarray}
\end{subequations}

By assumption,
$R_{L_i}=R_{R_i}=\sum_{j=1}^{m+1} X_j$ for
$1\leq i \leq m$. So $(\ref{eq:decode2_klinear})$ implies
\begin{subequations}\label{eq:code2_klinear}
\begin{eqnarray}
&& \gamma_{i, L}^{\prime} \beta_{m+1, i} = I \mbox{ for } 1\leq i \leq m, \label{code18}\\
&& \gamma_{i, L} \beta_{i}^{i} = I \mbox{ for } 1\leq i \leq m, \label{code19}\\
&& \gamma_{i, L}\beta_j^{i}+\gamma_{i,L}^{\prime}\beta_{j,i} = I \nonumber \mbox{ for } i, j=1, 2, \ldots, m, j\neq i, \label{code20}\\
&& \gamma_{i, R} = I \mbox{ for } 1\leq i \leq m, \label{code21}\\
&&\gamma_{i, R}^{\prime}\beta_{j, i} = I \nonumber \mbox{ for } 1\leq i \leq m, 1 \leq j \leq m+1, j\neq i. \label{code23}
\end{eqnarray}
\end{subequations}

All the coding matrices in (\ref{code18}), (\ref{code19}),
(\ref{code21}) and (\ref{code23}) are invertible
since the right hand side of the equations are the identity matrix.
Eq. $(\ref{code18})$ and $(\ref{code23})$ together imply 
\begin{eqnarray}
\gamma_{i,L}^{\prime} = \gamma_{i,R}^{\prime} \mbox{ for } 1\leq i \leq m. \label{code24}
\end{eqnarray}

By $(\ref{code23})$ and $(\ref{code24})$, we have 
\begin{eqnarray}
\gamma_{i,L}^{\prime}\beta_{j,i} = I \mbox{ for } 1\leq i,j \leq m, j\neq i. \label{code25}
\end{eqnarray}

By $(\ref{code20})$ and $(\ref{code25})$, we have 
\begin{eqnarray}
\gamma_{i,L} \beta_j^{i} = {\bf 0} \ \mbox{   for } 1\leq i,j \leq m, \ j\neq i,\nonumber 
\end{eqnarray}
where ${\bf 0}$ denotes the all-zero $\bl \times \bl$ matrix.
Since $\gamma_{i,L}$ is invertible by (\ref{code19}), $\beta_j^{i}=0$ for 
$1\leq i,j \leq m, \ j\neq i$. 
By (\ref{code19}), $\beta_i^{i}$ is an invertible matrix for $1\leq i \leq m$. 
So, we conclude that for every $i=1,\ldots,m$, $Y_{(z_i,t_{Li})}=
\beta_i^{i} X_i$.
This implies that for every $i=1,\ldots,m$, $z_i$ can recover
$X_i$, and thus the multiple-unicast network $\net_1$ is 
$\bl$-length vector linear solvable over $F$. This completes the proof of Part $(i)$.\\

\textbf{Proof of part $(ii)$:} \ Now we consider the case
 when nodes are allowed to do non-linear network coding, 
i.e., nodes can send any function of the incoming symbols 
on an outgoing edge. For the forward part, as in part (i), a code
for $\net_1$ can be extended to a code for $\net_2$.

Now we prove the ``only if'' part. Let us consider any network code for
$\net_2$ over $G$.
Without loss of generality, we assume
\begin{eqnarray}
Y_{(s_j, u_i)} & = & X_j
\mbox{ for } 1\leq i \leq m, 1\leq j \leq m+1, j \neq i, \nonumber\\
Y_{(s_i,w_i)} & = & Y_{(s_i, t_{R_i})} =  X_i \mbox{ for } 1\leq i \leq m, \nonumber \\
Y_{(v_i, t_{L_i})} & = & Y_{(v_i, t_{R_i})} = Y_{(u_i, v_i)} \mbox{ for } 1\leq i \leq m. \nonumber
\end{eqnarray}

Further, we assume that
\begin{subequations}\label{eq:nonlinear_code1}
\begin{eqnarray}
Y_{(u_i, v_i)} & = & f_{1}^{i}(X_1,\cdots,X_{i-1}, 
X_{i+1},\cdots,X_m,X_{m+1}) \nonumber \mbox{ for } 1\leq i \leq m, \nonumber \\
Y_{(z_i, t_{L_i})} & = & f_{2}^{i}(X_1,\cdots,X_{i-1}, 
X_{i},X_{i+1},\cdots,X_m) \nonumber \mbox{ for } 1\leq i \leq m, \nonumber 
\end{eqnarray}
\end{subequations}

and the decoded symbols at the terminals are 
\begin{subequations}\label{eq:nonlinear_decode}
\begin{eqnarray}
R_{R_i} & = & g_1^{i}(Y_{(v_i, t_{R_i})},Y_{(s_i, t_{R_i})}), \label{nonlinear_code7}\\
R_{L_i} & = & g_2^{i}(Y_{(v_i, t_{L_i})},Y_{(z_i, t_{L_i})}). \label{nonlinear_code6} 
\end{eqnarray}
\end{subequations}
Here all the symbols carried by the links are symbols from $G$.

We need to show that for every $i=1,
2,\ldots,m$, communicating the sum of the source symbols to the terminals
$t_{L_i}$ and $t_{R_i}$ is possible only if $f_2^{i}$ is an $1-1$ function
only of the symbol $X_i$ and is independent of the other variables.

By $(\ref{eq:nonlinear_decode})$, for every $i=1,2,\ldots,m$, the functions
$f_1^{i}$, $f_2^{i}$, $g_1^{i}$ and $g_2^{i}$ must satisfy the following conditions.
\begin{subequations}\label{eq:nonlinear_decode1}
\begin{eqnarray}
g_{1}^{i}(f_1^{i},X_i) & = & X_1+\cdots+X_{i-1}+X_{i}+X_{i+1}+\cdots +X_m +X_{m+1}, \label{nonlinear_code9} \\
g_{2}^{i}(f_1^{i},f_2^{i}) & = & X_1+\cdots+X_{i-1}+X_{i}+X_{i+1}+\cdots +X_m +X_{m+1}. \label{nonlinear_code8}
\end{eqnarray}
\end{subequations}

Now we prove the following claims for the functions
$g_{1}^{i}, g_{2}^{i}, f_1^{i}, f_2^{i}$ for $1 \leq i \leq m$.
\begin{claim}
\label{clm:f_1_bijective}
For every $i=1,2,\ldots,m$, $f_1^{i}$ is bijective on each 
variable $X_j, \mbox{ for } 1\leq j \leq m+1, j \neq i$, 
for any fixed values of the other variables.
\end{claim}
\textit{\ \ \ \ Proof:} \ 
Let us consider any $j \neq i$. For any fixed values of 
$\{X_k|k\neq j\}$, $(\ref{nonlinear_code9})$ implies that
$g_1^{i}(f_1^i(\cdot ,X_j, \cdot), \cdot)$ is 
a bijective function of $X_j$. This in turn implies that
$f_1^{i}$ is a bijective function of $X_j$ for any fixed values of the other
variables. 

\begin{claim}
\label{clm:g_1_bijective}
For every $i=1,2,\ldots,m$, $g_1^{i}(\cdot , \cdot)$ is bijective on each argument
for any fixed value of the other argument.
\end{claim}
\textit{\ \ \ \ Proof:} \
For any element of $G$, by claim 1, there exists
a set of values for $\{X_j|j\neq i\}$ so that the first argument
$f_1^i(\cdot)$ of $g_1^i$ takes that value.
For such a set of fixed values of $\{X_j|j\neq i\}$,
$g_1^i(\cdot ,X_i)$ is a bijective function of $X_i$ by (\ref{nonlinear_code9}).
Now, consider any $j\neq i$ and fix some values for $\{X_k|k\neq j\}$.
Again by (\ref{nonlinear_code9}), $g_1^i(f_1^i(\cdot ,X_j,\cdot ),\cdot)$
is a bijective function of $X_j$. This implies that $g_1^i$ is a bijective
function of its first argument for any fixed value of the second argument.

\begin{claim}
\label{clm:f_1_symmetric}
For every $i=1,2,\ldots,m$, $f_1^{i}$ is symmetric, i.e., interchanging the 
values of any two variables in its arguments does not change the value of
$f_1^{i}$.  
\end{claim}
\textit{\ \ \ \ Proof:} \
For some fixed values of all the arguments of $f_1^{i}$, suppose the value
of $f_1^{i}$ is $c_1$. We also fix the
value of $X_i$ as $c_2$. Suppose $g_1^i(c_1,c_2)=c_3$. Now we interchange the
values of the variables $X_j$ and $X_k$ where $j \neq i \neq k$.
Then it follows from Claim \ref{clm:g_1_bijective} and (\ref{nonlinear_code9}) that the value of $f_1^{i}$ must remain the same. 

\begin{claim}
\label{clm:f_2_bijective}
For every $i=1,2,\ldots,m$, $f_2^{i}$ is a bijective function of $X_i$
for any fixed values of $\{X_j|j=1,2,\ldots, m, j\neq i\}$.
\end{claim}
\textit{\ \ \ \ Proof:} \
For any fixed values of $\{X_j|j=1,2,\ldots, m, j\neq i\}$ and $X_{m+1}$,
$g_2^{i}(\cdot,f_2^i(\cdot,X_i,\cdot))$
is a bijective function of $X_i$ by equation $(\ref{nonlinear_code8})$.
This implies that $f_2^{i}$ is a bijective function of $X_i$ for
any fixed values of the other arguments.

\begin{claim}
\label{clm:g_2_bijective}
For every $i=1,2,\ldots,m$, $g_2^{i}(\cdot ,\cdot)$ is bijective on each argument
for any fixed value of the other argument.
\end{claim}

\textit{\ \ Proof:} \ 
By equation (\ref{nonlinear_code8}), $g_2^i(f_1^i(\cdot ,X_{m+1}),f_2^i(\cdot ))$ and
$g_2^i(f_1^i(\cdot),f_2^i(\cdot ,X_i, \cdot))$ are both bijective functions of $X_{m+1}$
and $X_i$ respectively for any fixed values of the omitted variables.
This implies that $g_2^i$ is a bijective function of the first and the
second argument for the other argument fixed.

Now to prove that for every $i=1,2,\ldots,m$, the value of
$f_2^{i}(X_1,\cdots, X_m)$ does not depend on $\{X_j|j=1,\ldots, m;j\neq i\}$,
it is sufficient to prove that for any set of fixed values 
$X_1=a_1,\ldots, X_m=a_m$, changing the value of $X_j$ ($j\neq i$)
to any $b_j \in G$ does not change the value of $f_2^i$. Let us assign $X_{m+1} = b_j$.
By (\ref{nonlinear_code8}), the value of $g_2^i$ does not change
by interchanging the values of $X_j$ and $X_{m+1}$. Also, the value
of $f_1^i$ does not change by this interchange by Claim \ref{clm:f_1_symmetric}.
So, by
Claim \ref{clm:g_2_bijective}, the value of $f_2^i$ also does not
change by this change of value of $X_j$ from $a_j$ to $b_j$.

Now using Claim \ref{clm:f_2_bijective}, it follows that $f_2^i$ is a bijective function
of only the variable $X_i$. This completes the proof of part (ii)
of the theorem. It is interesting to note that a somewhat similar technique
was independently used in \cite{langberg3} to prove that there is no non-linear solution
for the network $\spc_3$ (Fig.~\ref{fig:spc3}).

\textbf{Proof of part $(iii)$:} \ 
In the same way as in the proof of part (i), a $\bl$-length vector linear 
code of $\net_1^{\prime}$ can be extended to a $\bl$-length vector linear
code of $\net_2^{\prime}$.
Now we prove the ``only if'' part. Let us consider a $\bl$-length vector
linear code of $\net_2^{\prime}$. 
We assume that
for every $i=1,2,\ldots,m$, the edge $(w_i,s_i)$ carries a
linear combination
\begin{eqnarray}
Y_{(w_i,s_i)}=\mathop{\sum_{j=1}^{m}}
\beta_{i,j} X_{L_j}, \nonumber
\end{eqnarray}
where $\beta_{i,j} \in F^{\bl\times \bl}$.

Without loss of generality, we assume that
\begin{eqnarray}
& Y_{(u_j, s_i)} = & Y_{(v_j, u_j)} \mbox{ for } 1\leq j \leq m,
1\leq i \leq m+1, j\neq i, \nonumber \\
& Y_{(t_{L_j}, v_j)} = & Y_{(t_{L_j}, z_j)}  =  X_{L_j}
\mbox{ for } 1\leq j \leq m, \nonumber\\
\text{ and } & Y_{(t_{R_j}, s_j)}  = & Y_{(t_{R_j}, v_j)}  =  X_{R_j}
\mbox{ for } 1\leq j \leq m. \nonumber 
\end{eqnarray}

We further assume that
 \begin{subequations}\label{eq:reverse_code}
 \begin{eqnarray}
Y_{(v_j, u_j)} & = & \beta_{L,j}Y_{(t_{L_j}, v_j)}+ \beta_{R,j}Y_{(t_{R_j}, v_j)} \mbox{ for } 1\leq j \leq m. \label{reverse_code4}
\end{eqnarray}
\end{subequations}
and that for $1\leq i\leq m+1$, the decoded symbols $R_i$ at the
terminals $s_i$ are 
\begin{subequations}\label{eq:reverse_decode}
\begin{eqnarray}
R_{i} & = & \mathop{\sum_{j=1}^{m}}_{j\neq i}\gamma_{j,i} Y_{(u_j, s_i)}
+ \gamma_{i,i} Y_{(t_{R_i}, s_i)} 
+ \gamma_{i} Y_{(w_i, s_i)} \nonumber  1\leq i \leq m, \label{reverse_code8} \\
\text{ and } R_{m+1} & = & \mathop{\sum_{j=1}^{m}}\gamma_{j,m+1} Y_{(u_j, s_{m+1})}.
\label{reverse_code9}
\end{eqnarray}
\end{subequations}

Here all the coding coefficients and decoding coefficients 
are $\bl \times \bl$ matrices over $F$, and the message vectors $Y_{(.,.)}$
carried by the links are $\bl$-length vectors
over $F$.

By $(\ref{eq:reverse_code})$ and 
$(\ref{eq:reverse_decode})$, we have 
\begin{subequations}\label{eq:reverse_decode2}
\begin{eqnarray}
R_i & = & \mathop{\sum_{j=1}^{m}}_{j\neq i} \gamma_{j,i} (\beta_{L,j} X_{L_j}+\beta_{R,j} X_{R_j}) 
 + \gamma_{i,i} X_{R_i} + \gamma_{i}\mathop{\sum_{j=1}^{m}} \beta_{i,j} X_{L_j} 
 \mbox{ for } 1\leq i\leq m \label{reverse_code11}, \\
\mbox{and } R_{m+1} & = & \mathop{\sum_{j=1}^{m}} \gamma_{j,m+1}(\beta_{L,j} X_{L_j} + \beta_{R,j} X_{R_j}).\label{reverse_code12}
\end{eqnarray}
\end{subequations}

By assumption, for every $i=1,2,\ldots,m+1$, 
\begin{eqnarray}
&& R_i=\sum_{j=1}^{m} (X_{L_j}+X_{R_j}).\label{eq:sum}
\end{eqnarray}

By (\ref{eq:reverse_decode2}) and (\ref{eq:sum}), we have
\begin{subequations}\label{eq:reverse_decode3}
\begin{eqnarray}
&& \gamma_{i,i} =  I \mbox{ for } 1\leq i \leq m, \label{reverse_code13} \\
&& \gamma_{j,m+1}\beta_{L,j} = \gamma_{j,m+1}\beta_{R,j} = I \mbox{ for }1\leq j \leq m, \label{reverse_code14}\\
&& \gamma_{j,i} \beta_{R,j} = I \mbox{ for } 1\leq i,j \leq m, j\neq i, \label{reverse_code15}\\
&& \gamma_{j,i} \beta_{L,j}+\gamma_i \beta_{i,j} = I \mbox{ for } 1\leq i,j \leq m, j\neq i, \label{reverse_code16}\\
&& \gamma_i \beta_{i,i} = I \mbox{ for } 1 \leq i \leq m.\label{reverse_code17}
\end{eqnarray}
\end{subequations}

All the coding matrices in (\ref{reverse_code13}), 
(\ref{reverse_code14}), (\ref{reverse_code15}) and  
(\ref{reverse_code17}) are invertible since the right hand side of the
equations are the identity matrix. Equation (\ref{reverse_code14}) implies
that $\beta_{L,j} = \beta_{R,j}$ for $1\leq j\leq m$. By this, (\ref{reverse_code15}) implies $\gamma_{j,i} \beta_{L,j} = I$. Then (\ref{reverse_code16}) 
implies
\begin{eqnarray}
\gamma_i \beta_{i,j} = {\bf 0} \ \mbox{ for } 1\leq i,j \leq m, j\neq i, \nonumber 
\end{eqnarray}
where ${\bf 0}$ is the $\bl \times \bl$ all-zeros matrix.

Since $\gamma_i$ is an invertible matrix for $1 \leq i \leq m$ by (\ref{reverse_code17}), we have 
\begin{eqnarray}
\beta_{i,j} = {\bf 0} \ \mbox{ for } 1\leq i,j \leq m, j\neq i. \nonumber 
\end{eqnarray}

Also from (\ref{reverse_code17}), $\beta_{i,i}$ is invertible. So, for every $i=1,2,\ldots,m$, the edge $(w_i,s_i)$ carries a nonzero scalar multiple
of $X_{L_i}$, which implies that the reverse multiple-unicast 
network $\net_1^{\prime}$ is solvable over $F$.
This completes the proof of part (iii).

\textbf{Proof of part $(iv)$:} 
First, similar to part (iii), any code over $G$ for $\net_1^\prime$
can be extended to a code over $G$ for $\net_2^\prime$.

Now we prove the converse. Consider any network code over $G$
for $\net_2^\prime$ where 
the terminal $s_i$ computes 
\begin{eqnarray}
&& R_i = \sum_{i=1}^m (X_{L_i} + X_{R_i}) \hspace{2mm} \mbox{for }
i=1,2,\ldots, m+1. \nonumber 
\end{eqnarray}

Without loss of generality, we assume that
\begin{subequations}\label{eq:reverse_nonlinear_code}
\begin{eqnarray}
&&Y_{(u_i, s_j)} = Y_{(v_i, u_i)} \mbox{ for } 1\leq i \leq m, 
\,\,\,1\leq j \leq m+1, i \neq j. \nonumber \\
&& Y_{(t_{R_i}, v_i)}  = Y_{(t_{R_i}, s_i)}  =  X_{R_i} \mbox{ for } 1\leq i \leq m, \label{reverse_nonlinear_code0} \nonumber \\
&& Y_{(t_{L_i}, z_i)}  = Y_{(t_{L_i}, v_i)}  =  X_{L_i} \mbox{ for } 1\leq i \leq m, \nonumber 
\end{eqnarray}

We further assume that
\begin{eqnarray}
Y_{(v_i, u_i)} & = & f_i(X_{L_i},X_{R_i}) \mbox{ for } 1\leq i \leq m, \nonumber \\
Y_{(w_i, s_i)} & = & f_i^{\prime}(X_{L_1},X_{L_2},\cdots,X_{L_m}), 
\mbox{ for } 1\leq i \leq m, \nonumber 
\end{eqnarray}
\end{subequations}

and the decoded symbols are
\begin{subequations}\label{eq:reverse_nonlinear_decode}
\begin{eqnarray}
& R_i  & = g_{i}(Y_{(w_i,s_i)},Y_{(u_1, s_i)}, \cdots, Y_{(u_{i-1},s_{i})}, Y_{(u_{i+1},s_{i})}, \cdots, Y_{(u_{m},s_{i})}, Y_{(t_{R_i},s_i)})  
\label{reverse_nonlinear_code6}
\end{eqnarray}
for $1\leq i \leq m$; and
\begin{eqnarray}
& R_{m+1} & = g_{m+1}(Y_{(u_1, s_{m+1})}, \cdots, Y_{(u_{m},s_{m+1}))}. \label{reverse_nonlinear_code7}
\end{eqnarray}
\end{subequations}

Now we state some claims which can be proved using similar arguments
as in the proof of part (ii) of the theorem. We omit the proof of these
claims.

1. As a function of the variables $X_{L_i}, X_{R_i}$; $i=1,2,\ldots, m$,
$g_{m+1}$ is bijective in each variable for fixed values of the other
variables.

2. The function $g_{m+1}$ is bijective in each variable $Y_{(u_i, s_{m+1})}$ for fixed
values of the other variables.

3. For $i=1,2,\ldots, m$, $f_i(X_{L_i}, X_{R_i})$ is a bijective function
of each variable for a fixed value of the other variable.

4. For $i=1,2,\ldots, m$, $f_i(X_{L_i}, X_{R_i})$ is symmetric on its
arguments.

5. For $i=1,2,\ldots, m$, $g_i$ is bijective on each of its arguments
for fixed values of the other arguments.

First by (\ref{reverse_nonlinear_code6}), $g_i(f_i^\prime (\cdot,X_i, \cdot), \cdots)$, and
thus $f_i^\prime (\cdot,X_i, \cdot)$, is a bijective function of $X_i$ for
any fixed values of the other variables.
Now we prove that $f_i^{\prime}$ is a function of only $X_{L_i}$,
and it is independent of the other variables. Fix a $k\neq i$.
It is sufficient
to prove that for any fixed values of $X_{L_j}, j=1,2,\ldots, m, j\neq k$,
the value of $f_i^{\prime}$ does not change if the value of $X_{L_k}$ is 
changed from, say, $a$ to $b$.
Let us fix $X_{L_k}=a$. Let us further fix $X_{R_k}=b$ and the variables
$X_{R_j}$ for $j\neq k$ to arbitrary values. Now, by interchanging
the values of $X_{L_k}$ and $X_{R_k}$, the value of $g_i$ does not
change, since the sum of the variables does not change. Further,
all the arguments of $g_i$ other than $f_i^\prime$ does not change
since $f_k$ is symmetric on its arguments. So by claim 5, the
value of $f_i^\prime$ also does not change by this interchange. But
this means that the value of $f_i^\prime$  does not change by the
change of the value of $X_{L_k}$ from $a$ to $b$. This completes the proof of
part (iv).

\section{Proof of Theorem \ref{thm:main}}
\label{app1}

Define $m = p_1p_2\ldots p_l + 2$, where the empty product is defined
to be $1$.
We prove that the network $\spc_m$ satisfies the condition in the theorem.
First we prove that if, for any $n$, it is possible to communicate the
sum of the sources by $n$-length vector linear network coding over a 
field $F$ to all the terminals in $\spc_m$, then the characteristic of $\F$
must be from $\cP$. For $i = 1,2,\ldots, m$, the message vector generated
by the source $s_i$ is denoted by $X_i \in F^\bl$.

Without loss of generality, we assume that
\begin{subequations}\label{eq:smcode}
\begin{eqnarray}
Y_{(s_i, t_j)} & = & X_i 
\mbox{ for } 1\leq i,j \leq m-1, i\neq j,\nonumber \\
Y_{(s_i, u_i)} & = & X_i, 
\mbox{ for } 1\leq i \leq m-1,\nonumber \\
 Y_{(s_m, u_i)}& = & X_m 
\mbox{ for } 1\leq i \leq m-1,\nonumber \\
\text{and } Y_{(v_i, t_i)} & = & Y_{(v_i, t_m)} = Y_{(u_i, v_i)} \mbox{ for } 1\leq i \leq m-1. \nonumber
\end{eqnarray}
Let us further assume that
\begin{eqnarray}
Y_{(u_i, v_i)} & = & \beta_{i,1} Y_{(s_i, u_i)} + \beta_{i,2} Y_{(s_m, u_i)} \mbox{ for } 1\leq i \leq m-1\label{smcode4} 
\end{eqnarray}
\end{subequations}
and the vectors computed at the terminals are 
\begin{subequations}\label{eq:smdecode}
\begin{eqnarray}
R_i & = & \mathop{\sum_{j=1}^{m-1}}_{j\neq i} \gamma_{j,i} Y_{(s_j, t_i)}
+ \gamma_{i,i} Y_{(v_i, t_i)} \mbox{ for } 1\leq i \leq m-1,\label{smcode5} \\
\text{ and }R_m & = & \sum_{j=1}^{m-1} \gamma_{j,m} Y_{(v_j, t_m)}. \label{smcode6}
\end{eqnarray}
\end{subequations}
Here all the coding and decoding coefficients $\beta_{i,j},
\gamma_{i,j}$ are $\bl \times \bl$ matrices over $\F$, and the message vectors $X_i$
and the messages $Y_{(.,.)}$ carried by the links are length-$\bl$ vectors
over $\F$.

By assumption, each terminal decodes the sum of all the sources. That is,
\begin{eqnarray}
R_i = \sum_{j=1}^m X_j \mbox{ for } i=1,2,\ldots, m.  \label{eq:sum1}
\end{eqnarray}
 
From (\ref{eq:smcode}) and (\ref{eq:smdecode}), we have
\begin{eqnarray}
R_i = && \mathop{\sum_{j=1}^{m-1}}_{j\neq i}\gamma_{j,i}X_j
+ \gamma_{i,i}\beta_{i,1} X_i + \gamma_{i,i}\beta_{i,2}X_m \,\, \text{ for } i=1,2,\ldots, m-1 \label{eq:decode1}
\end{eqnarray}
and 
\begin{equation}
R_m  =  \sum_{j=1}^{m-1}\gamma_{j,m}\beta_{j,1}X_j
+ \left(\sum_{j=1}^{m-1}\gamma_{j,m}\beta_{j,2}\right) X_m. \label{eq:decode2}
\end{equation}
Since (\ref{eq:sum1}) is true for all values of $X_1, X_2, \ldots, X_m \in \F^\bl$,
(\ref{eq:decode1}) and (\ref{eq:decode2}) imply
\begin{eqnarray}
&&\gamma_{j,i} = \iden \mbox{ for } 1\leq i,j \leq m-1, i\neq j, \label{sol1} \\
&&\gamma_{i,i} \beta_{i,1} = \iden \mbox{ for } 1\leq i\leq m-1, \label{sol2} \\
&&\gamma_{i,i}\beta_{i,2} = \iden \mbox{ for } 1\leq i\leq m-1,  \label{sol3} \\
&&\gamma_{i,m}\beta_{i,1} = \iden \mbox{ for } 1\leq i\leq m-1,  \label{sol4} \\
&&\sum_{i=1}^{m-1} \gamma_{i,m}\beta_{i,2} = \iden,  \label{sol5} 
\end{eqnarray}
where $\iden$ denotes the $\bl \times \bl$ identity matrix over $\F$.
All the coding matrices in (\ref{sol1})--(\ref{sol4}) are invertible since the right hand side of the equations
are the identity matrix. Equations (\ref{sol2}) and (\ref{sol3}) imply
$\beta_{i,1} = \beta_{i,2}$ for $1\leq i\leq m-1$.
So, (\ref{sol5}) gives
\begin{eqnarray}
\sum_{i=1}^{m-1} \gamma_{i,m}\beta_{i,1} = \iden. \nonumber 
\end{eqnarray}
Now, using (\ref{sol4}), we get
\begin{eqnarray}
&& \sum_{i=1}^{m-1} \iden = \iden,  \nonumber \\
&\Rightarrow & (m-1) \iden = \iden, \nonumber \\
&\Rightarrow & (m-2) \iden = {\bf 0}. \nonumber
\end{eqnarray}
This is true if and only if the characteristic of $\F$ divides
$m-2$. So, the sum of the sources can be communicated in $\spc_m$ by
vector linear network coding only if the characteristic of $\F$
belongs to $\cP$.

Now, if the characteristic of $\F$ belongs to $\cP$, then
for any block length $\bl$, in particular for scalar network coding for
$\bl = 1$, every coding matrix in (\ref{smcode4}) and (\ref{eq:smdecode})
can be chosen
to be the identity matrix. The terminals then recover the sum of
the sources.  This completes the proof.

\section{Proof of Theorem \ref{thm:main_cofinite}}
\label{app2}

Consider the network $\spc_m^*$ shown in Fig. \ref{fig:spcm*} for
$m = p_1p_2\ldots p_l + 2$. We will show that this
network satisfies the condition of the theorem.
We first prove that if it is possible to communicate the sum of the
source messages using vector linear network coding over a field
$\F$ to all the terminals in $\spc_m^{*}$, then the characteristic
of $\F$ must not divide $m-2$.
For $i = 1,2,\ldots, m$, let the message vector generated by the
source $s_i$ be denoted by $X_i \in \F^\bl$. Each terminal $t_i$
computes a linear combination $R_i$ of the received vectors.

Without loss of generality, let us assume that
\begin{subequations}\label{eq:smcode_new}
\begin{eqnarray}
Y_{(s_i, t_i)} & = &  X_i
\hspace*{0mm} \mbox{ for } 1\leq i \leq m-1,\label{code1_new} \\
Y_{(s_i, u_j)} & = &  X_i 
\hspace*{0mm} \mbox{ for } 1\leq i,j \leq m-1, i\neq j,\label{code2_new}\\
\text{ and } Y_{(v_i, t_i)} & = & Y_{(v_i, t_m)} = Y_{(u_i, v_i)} \mbox{ for } 1\leq i \leq m-1 .
\end{eqnarray}
Further, suppose
\begin{eqnarray}
Y_{(u_i, v_i)} & = & \mathop{\sum_{j=1}^{m-1}}_{j\neq i} \beta_{j,i} Y_{(s_j, u_i)}
\hspace*{0mm} \mbox{ for } 1\leq i \leq m-1,\label{code4_new}
\end{eqnarray}
\end{subequations}
\vspace*{-5mm}
\begin{subequations}\label{eq:smdecode_new}
\begin{eqnarray}
R_i & = & \gamma_{i,1} Y_{(s_i, t_i)}
+ \gamma_{i,2} Y_{(v_i, t_i)} \mbox{ for } 1\leq i \leq m-1,\label{code5_new} \\
\text{and } R_m & = & \sum_{j=1}^{m-1} \gamma_{j,m} Y_{(v_j, t_m)}. \label{code6_new}
\end{eqnarray}
\end{subequations}
Here all the coding coefficients $\beta_{i,j}, \gamma_{i,j}$ are $\bl \times \bl$ matrices over $\F$.

By assumption, each terminal decodes the sum of all the source messages. That is,
\begin{eqnarray}
R_i = \sum_{j=1}^{m-1} X_j \mbox{ for } i=1,2,\ldots, m .
\label{eq:sum1_new}
\end{eqnarray}

From (\ref{eq:smcode_new}) and (\ref{eq:smdecode_new}), we have
\begin{eqnarray}
R_i & = &  \mathop{\sum_{j=1}^{m-1}}_{j\neq i}\gamma_{i,2}\beta_{j,i}X_j+ \gamma_{i,1} X_i 
\mbox{ for } i=1,2,\ldots, m-1, \label{eq:decode1_new} \\
\text{and } R_m  &=&  \sum_{i=1}^{m-1}\gamma_{i,m}\left(\mathop{\sum_{j=1}^{m-1}}_{j\neq i}\beta_{j,i}X_j\right) \nonumber \\
& = & \sum_{j=1}^{m-1}\left(\mathop{\sum_{i=1}^{m-1}}_{i\neq j}\gamma_{i,m}\beta_{j,i}\right)X_j.\label{eq:decode2_new} 
\end{eqnarray}
Since (\ref{eq:sum1_new}) is true for all values of $X_1, X_2, \ldots, X_m \in \F^\bl$,
(\ref{eq:decode1_new}) and (\ref{eq:decode2_new}) imply
\begin{eqnarray}
&& \gamma_{i,2}\beta_{j ,i} = \iden \mbox{ for } 1\leq i, j \leq m-1, i\neq j, \label{sol1_new} \\
&& \gamma_{i,1} = \iden \mbox{ for } 1\leq i\leq m-1, \label{sol2_new} \\
&& \mathop{\sum_{i=1}^{m-1}}_{i\neq j}\gamma_{i,m}\beta_{j,i} = \iden \mbox{ for } 1\leq j\leq m-1.  \label{sol3_new}
\end{eqnarray}
All the coding matrices in (\ref{sol1_new}), (\ref{sol2_new}) are invertible since the right hand side of the equations are the identity matrix. Equation (\ref{sol1_new}) implies
$\beta_{j,i} = \beta_{k,i}$ for $1\leq i, j, k\leq m-1$, $j\neq i\neq k$.
So, let us denote all the equal co-efficients $\beta_{j,i}; 1\leq j \leq m-1,
j\neq i$ by $\beta_i$. Then (\ref{sol3_new}) can be rewritten as
\begin{eqnarray}
\mathop{\sum_{i=1}^{m-1}}_{i\neq j}\gamma_{i,m}\beta_{i} = \iden \mbox{ for }
1\leq j\leq m-1.
\label{sol4_new}
\end{eqnarray}
Equating the left hand side of (\ref{sol4_new}) for two values of
$j$ (say, $j$ and $k$), we have
\begin{eqnarray}
\gamma_{j,m}\beta_{j} = \gamma_{k,m}\beta_{k}  \mbox{ for } 1\leq j, k\leq m-1. \nonumber
\end{eqnarray}
Then (\ref{sol4_new}) gives
\begin{eqnarray}
(m-2)\gamma_{1,m}\beta_{1}  = \iden . 
  \label{sol6_new}
\end{eqnarray}

Equation (\ref{sol6_new}) implies that $(m-2)$ is invertible in $F$, that is,
the characteristic of $F$ does not devide $(m-2)$.
So the sum of the source messages can be communicated in $\spc_m^{*}$ by $n$-length vector
linear network coding over $F_q$ only if the characteristic of $\F$ does not
divide $(m-2)$.

Now, if the characteristic of $\F$ does not divide $(m-2)$, then for any block
length $\bl$,
every coding matrix in (\ref{code4_new}) and (\ref{code5_new}) can be
chosen to be the identity matrix, and $\gamma_{j,m} = (m-2)^{-1} I$ in
(\ref{code6_new}). The terminals $t_{1},t_{2},\cdots,
t_{m}$ then recover the sum of
the source messages. This completes the proof.

\begin{biographynophoto}{Brijesh Kumar Rai} (S'09-M'11)
received the B.E. degree in Electronics
and Communication Engineering from F.E.T., R.B.S College, Bichpuri,
Agra, India, in 2001, the M.Tech. degree in Electrical Engineering
from Indian Institute of Technology Kanpur, Kanpur, India, in 2004,
and the Ph.D. degree in Electrical Engineering from Indian Institute
of Technology Bombay, Mumbai, India, in 2010.

Currently, he is an Assistant Professor at the Department of
Electronics and Electrical Engineering, Indian Institute of Technology
Guwahati, Guwahati, India. From September 2010 to May 2011, he was a
Postdoctoral Fellow at Network and Computer Science Department -
INFRES, T\'{e}l\'{e}com ParisTech, Paris, France. His research interests
include information theory, communications, coding theory and network
coding.

Prof. Rai has been selected for Microsoft Young Outstanding Faculty
Award 2011-2012 at the Department of Electronics and Electrical
Engineering, Indian Institute of Technology Guwahati, Guwahati, India.
\end{biographynophoto}

\begin{biographynophoto}{Bikash Kumar Dey} (S'00-M'04)
received his B.E. degree in Electronics and Telecommunication Engineering from Bengal Engineering College, Howrah, India, in 1996. He received his M.E. degree in Signal Processing and Ph.D. in Electrical Communication Engineering from the Indian Institute of Science in 1999 and 2003
respectively.

From August 1996 to June 1997, he worked at Wipro Infotech Global R\&D. In February 2003, he joined Hellosoft India Pvt. Ltd. as a Technical Member. In June 2003, he joined the International Institute of Information Technology, Hyderabad, India, as Assistant Professor. In May 2005, he joined the Department of Electrical Engineering of Indian Institute of Technology Bombay where he works as Associate Professor. His research interests include information theory, coding theory, and network coding.

He was awarded the Prof. I.S.N. Murthy Medal from IISc as the best M.E. student in the Department of Electrical Communication Engineering and Electrical Engineering for 1998-1999 and Alumni Medal for the best Ph.D. thesis in the Division of Electrical Sciences for 2003-2004.
\end{biographynophoto}

\end{document}